\documentclass[pdftex,a4paper,parskip,fleqn]{scrartcl}
\usepackage[ansinew]{inputenc}
\usepackage[T1]{fontenc}
\usepackage{graphicx}
\usepackage{struktex}
\usepackage{listings}
\usepackage{mathcomp}
\usepackage{mathtools}
\usepackage{amsmath, amsthm, amssymb}
\usepackage{epstopdf}
\usepackage{hyperref}
\usepackage[linesnumbered, ruled, vlined]{algorithm2e}
\usepackage{enumitem}
\usepackage{subcaption}
\usepackage{wrapfig}
\usepackage{caption}
\usepackage{xcolor}
\usepackage[normalem]{ulem}
\usepackage[babel=true]{csquotes}
\usepackage{cleveref}
\usepackage{hypcap}
\newtheorem{theorem}{Theorem}

\newtheorem{lemma}[theorem]{Lemma}
\newtheorem{definition}{Definition}

\def\vE{\vec{E}}
\def\vG{\vec{G}}
\def\eqd{\equiv_\delta}
\def\eqk{\equiv_{\delta^*_k}}
\def\eqkp{\equiv_{\delta^*_{k+1}}}
\def\neqkp{\not\equiv_{\delta^*_{k+1}}}
\def\neqd{\not\equiv_\delta}
\def\neqk{\not\equiv_{\delta_k^*}}
\def\N{\mathbb{N}}
\def\delk{\delta^*_k}
\def\delkp{\delta^*_{k+1}}
\def\eqc{\equiv_{\delta^c}}

\begin{document}

\title{Finding Euler Tours in One Pass in the W-Streaming Model with O(n~log(n)) RAM}
\author{Christian Glazik \and Jan Schiemann \and Anand Srivastav}
\date{\small%
Department of Computer Science\\
Kiel University\\
Christian-Albrechts-Platz 4\\
24118 Kiel, Germany\\
\texttt{\{cgl,jasc,asr\}@informatik.uni-kiel.de}}
\maketitle

\begin{abstract}
\textbf{Abstract:}
  We study the problem of finding an Euler tour in an undirected graph $G$
  in the W-Streaming model with $\mathcal O(n\text{ polylog}(n))$ RAM,
  where $n$ resp.\ $m$ is the number of nodes resp.\ edges of $G$.
  Our main result is the first one pass W-Streaming algorithm computing an 
Euler 
tour of
  $G$ in the form of an edge successor function with only $\mathcal O(n 
\log(n))$
  RAM which is optimal for this setting (e.g., Sun and Woodruff (2015)).  The
  previously best-known result in this model is implicitly given by Demetrescu
  et al.\ (2010) with the parallel algorithm of Atallah and Vishkin (1984)
  using $\mathcal O(m/n)$ passes under the same RAM limitation.
  For graphs with $\omega(n)$ edges this is non-constant.

  Our overall approach is to partition the edges into edge-disjoint cycles and 
  to merge the cycles until a single Euler tour
  is achieved.
  Note that in the W-Streaming model such a merging is far from being obvious
  as the limited RAM allows the processing
  of only a constant number of cycles at once.
  This enforces us to merge cycles
  that partially are no longer present in RAM\@.
  Furthermore,  the successor of an edge cannot be changed
  after the edge has left RAM\@.
  So,
  we steadily have to output edges and their designated successors,
   not knowing the appearance of edges and cycles yet to come.
   We solve this problem with a special edge swapping
   technique, for which two certain edges per node
   are sufficient to merge tours without having
  all of their edges in RAM.
  Mathematically, this is controlled
  by structural results on the space of certain equivalence classes
  corresponding to cycles and the characterization of associated successor 
functions.
  For example, we give conditions under which
  the swapping of edge successors
  leads to a merging of equivalence classes.
  The mathematical methods of our analysis
  %are new and
  might be of independent interest
  for other routing problems in streaming models
\par
\end{abstract}

% % REQUIRED
% \begin{AMS}
%   68Q25, 68R10, 68U05
% \end{AMS}

\section{Introduction}

For the processing of large graphs,
the \emph{graph streaming} or \emph{semi streaming} model introduced by
Feigenbaum et al.~\cite{Feigenbaum:2005:GPS:1132633.1132638}
has been studied extensively over the last decade.
In this model, a graph with $n$ nodes and $m$ edges
is given as a stream of its edges.
Random-access memory (RAM, also called internal memory)
is restricted to $\mathcal O(n \text{ polylog}(n))$ edges at a time,
see, e.g., the survey~\cite{McGregor}
for a detailed introduction.
In consequence, the model cannot be applied to problems
where the size of the solution
exceeds this amount of memory.
In the Euler tour problem,
we are looking for a closed trail in an undirected graph
such that each edge is visited exactly once.
Since the size of an Euler tour is $m$,
which might even be $\Theta(n^2)$,
we need a relaxation of the model
that allows us to store the output separated from the RAM\@.

\subsection{Previous Work on W-Streaming}
The \emph{W-Streaming model} introduced by Demetrescu et
al.~\cite{Dem09} is a relaxation of the classical streaming
model. 
At each pass, an output stream is written
which then becomes the input stream of the next pass.  In~\cite{Dem09}, a
trade-off between internal memory and streaming passes is shown for undirected
connectivity and the single-source shortest path problem in directed graphs.
The One-Pass-Model plays a special role, since in this case the
writing to the stream is only for output reasons,
because the stream is only processed once.
This is particularly interesting regarding problems with solutions which do 
not fit in RAM.
The W-Streaming model originated as a more restrictive alternative to the
\emph{StrSort model} introduced by Aggarwal et al.~\cite{Ruhl,Aggarwal}.

Finding an Euler tour in trees
has been studied in multiple papers (e.g.,~\cite{Dem10}),
but to the best of our knowledge the general Euler tour problem
has hardly been considered in a streaming model so far.
However, there are some general results for transferring PRAM algorithms to the 
W-Streaming model.
Atallah and Vishkin~\cite{Atallah} presented a PRAM algorithm for finding Euler 
tours,
using $\mathcal O(\log(n))$ time and $n+m$ processors.
Transferred to the W-Streaming model with the methods from~\cite{Dem10}, this 
algorithm
computes an Euler tour in the form of a bijective successor function 
within $p=\mathcal O(m\text{ polylog}(n)/s)$ passes,
where $s$ is the RAM-capacity.
For a RAM size of $\mathcal O(n\text{ polylog}(n))$, this translates to
$\Omega(m/n)$ passes which for any $m=\omega(n)$ is non-constant.  Furthermore,
Sun and Woodruff~\cite{Sun} showed that a one-pass streaming
algorithm for verifying whether a graph is Eulerian needs $\Omega(n \log(n))$
RAM\@. This implies that a one pass W-streaming algorithm for finding an Euler 
tour
with less RAM does not exist and therefore justifies our choice of the RAM size.

\subsection{Our Contribution}
We present the W-Stream algorithm \textsc{ Euler-Tour}
for finding an Euler tour in a graph in form of a bijective successor function 
or stating that the graph is not Eulerian, using only one pass and 
$\mathcal O(n \log(n))$ bits of RAM\@.
This is not only a significant improvement over previous results,
but is in the view
of the lower bound of Sun and Woodruff~\cite{Sun}
the first optimal algorithm in this setting.
Usually, the W-Streaming model is restricted to sub-linear internal memory but 
in our case the output stream is used solely for storing the solution 
which needs $\Omega(m)$ memory.
As in~\cite{Atallah}, our algorithm outputs the Euler tour 
as a successor function that for every edge gives the following edge of the 
tour.
Atallah and Vishkin~\cite{Atallah} find edge disjoint tours (in our case 
cycles) 
and connect them by 
pairwise swapping the successor edges of suitable edges.
This idea is easy to implement without memory restrictions
but the implementation gets distinctly more complicated
with limited memory space:
We cannot store all cycles in RAM\@.
Therefore, we have to output edges and their successors before finding resp.\ 
processing
all cycles.
Our idea is to keep specific edges of some cycles in RAM along with additional 
information so
that we are able to merge following cycles regardless of their appearance with 
already processed tours
which likely are no longer present in RAM\@.
\par
We develop a mathematical foundation by partitioning the edges into
equivalence classes induced by a given bijective successor function and prove
structural properties that allow to iteratively change this function 
on a designated set of edges so that the modified function
is still bijective. Translated to graphs this is a tour
merging process.  This mathematical approach is quite general and might be
useful in other routing scenarios in streaming models.

\subsection{Organization of the Article}  
In Section~\ref{sec:preliminaries} we give some basic definitions.
The main techniques of our algorithm are described in Section~\ref{sec:idea} 
in an intuitive manner.
Section~\ref{sec:algorithm} contains the pseudo
code of the algorithm.  
In the analysis in Section~\ref{sec:analysis},
we show the
connection of the concepts of Euler tours and successor functions
and then show that the required RAM 
of the algorithm does not
exceed $\mathcal O(n \log(n))$ and that the output actually depicts an Euler
tour (Theorem~\cref{finalthm}). 

\section{Preliminaries}\label{sec:preliminaries}

Let $\N:=\{1,2,\ldots\}$ denote the set of natural numbers.
For $n\in\N$ let $[n]:=\{1,\ldots,n\}$.
In the following, we consider a graph $G=(V,E)$
where $V$ denotes the set of nodes and $E$ the set of (undirected) edges.
A \emph{trail} in $G$ is a finite sequence $T=(v_1,\ldots,v_\ell)$ of nodes of 
$G$
with $\{v_i,v_{i+1}\}\in E$ and $v_i=v_j$ implies 
$v_{i+1}\notin\{v_{j-1},v_{j+1}\}$
for all $i\in\{1,\ldots,\ell-1\}$ and $j\in\{2,\ldots,\ell-1\}$ with $i\neq j$.
The \emph{length} of $T$ is $\ell-1$.
The (directed) edge-set of $T$ is $E(T):=\{(v_i,v_{i+1})|i\in[\ell-1]\}$.
We also write $e\in T$ instead of $e\in E(T)$.
For a directed edge $e$ we denote by $e_{(1)}$ its first and by $e_{(2)}$ its 
second component.
A trail $T=(v_1,\ldots,v_\ell)$ with $v_1=v_\ell$ is called a \emph{tour}.
In tours, we usually do not care about starting point and end point,
so we slightly abuse the notation and write $v_{i+k}$ resp. $v_{i-k}$ for any 
$k\in\N$,
identifying $v_0:=v_\ell$ and $v_{\ell+1}:=v_2$ and so on.
If additionally $v_i\neq v_j$ holds for all $i,j\in [\ell-1],i\neq j$
(and $\ell\geq 3$), we call $T$ a \emph{cycle}.
An \emph{Euler tour} of $G$ is a tour $T$ with $E(T)=E$.
Since in the streaming model the graph is represented as a set of edges, we 
often use the edges for the depiction of tours.
With $e_i:=\{v_i,v_{i+1}\}$ for all $i \in [l-1]$, $T$ can be written as 
$T=(e_1,\ldots,e_l)$.
Here, we also use the slightly abusive index notation.
Note that for the tour $T$ the edges are distinct.
For $i \in [l]$, we call $e_{i+1}$ the \emph{successor edge of $e_i$} in tour 
$T$.
Our algorithm outputs an Euler tour $T=(v_1,\ldots,v_{|E|},v_1)$ in form of 
  a \emph{successor function}, i.e.,
  for every $i \in [|E|]$, we output the triple $(v_i,v_{i+1},v_{i+2})$, where 
  $\{v_{i+1},v_{i+2}\}$ is the successor edge of $\{v_i,v_{i+1}\}$ in $T$.

\section{Idea of the Algorithm}\label{sec:idea}

As the analysis of our algorithm is quite involved, in this section
we try to explain the new algorithmic idea and where the mathematical analysis 
is required.
First we explain how merging of subtours can be accomplished without RAM 
limitation
clarifying why this does not work in W-streaming. Thereafter we explain our 
merging technique
and its locality and RAM efficiency.

\subsection{Subtour merging in unrestricted RAM}
Recall that an Euler tour will be presented by giving for every edge the 
corresponding
successor edge in the tour.
Let $G=(V,E)$ be an Eulerian graph and $T, T'$ be edge-disjoint tours in $G$.
The tour induces an orientation of the edges in a canonical way.
If $T$ and $T'$ have a common node $v$, it is easy to merge them to a 
single tour:
$T$ has at least one in-going edge $(u,v)$ with a successor edge $(v,w)$,
and $T'$ has at least one in-going edge $(u',v)$ with a successor edge 
$(v,w')$.
By changing the successor edge of $(u,v)$ from $(v,w)$ to $(v,w')$
and the successor edge of $(u',v)$ to $(v,w)$, we get a tour containing all 
edges of $T \cup T'$ (see Figure~\ref{fig:TwoTours}).
\begin{figure}
  \centering
  \fbox{
    \includegraphics[width=0.4\linewidth]{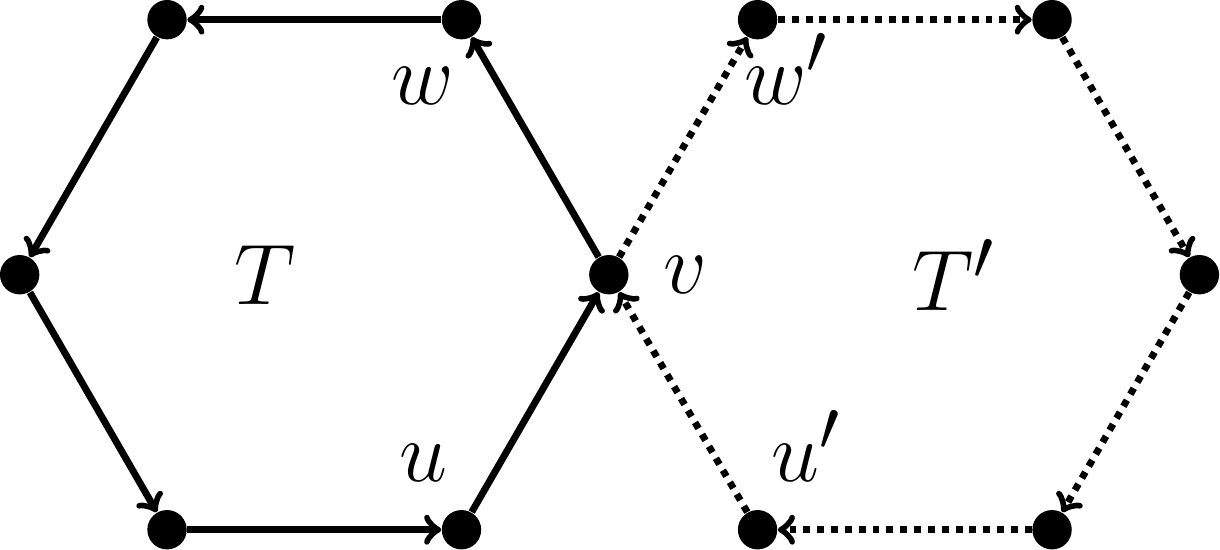}
  }
  \fbox{
    \includegraphics[width=0.4\linewidth]{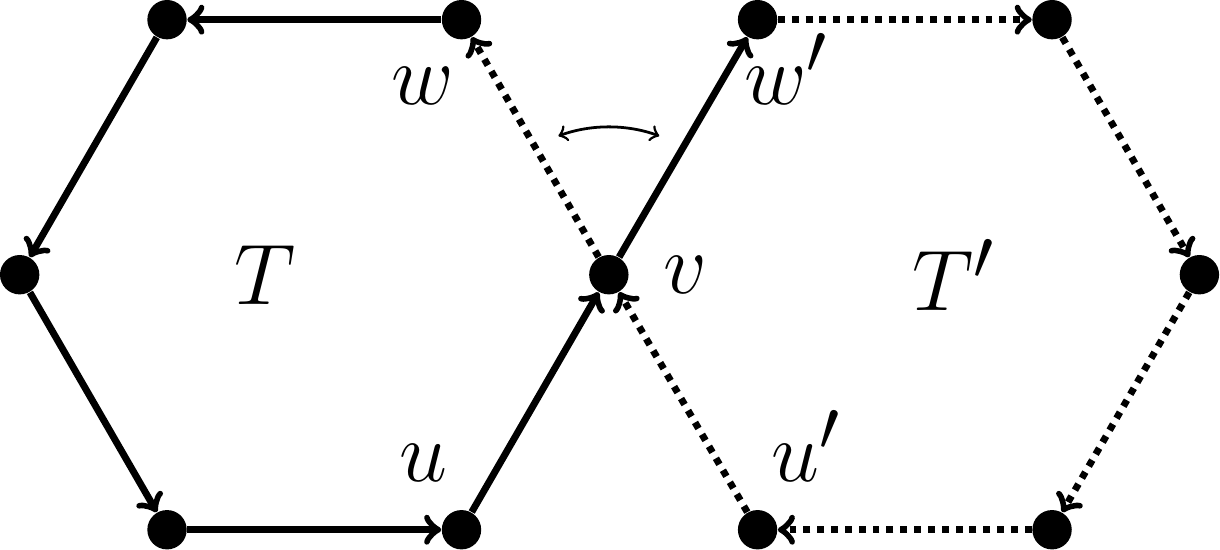}
  }
  \caption{Merging tours by swapping two edges}\label{fig:TwoTours}
\end{figure}
The same principle can be applied when merging more than two tours at once.
When we have a tour $T$ and tours $T_1,\ldots,T_k$, $k \in \mathbb N$, 
such that $T, T_1,\ldots,T_k$ are pairwise edge-disjoint
and for every $j \in [k]$ there is a common node $v_j$ of $T$ and $T_j$,
switching the successor edges of two in-going edges per node $v_j$ as described 
above 
creates a tour containing the edges of $T \cup T_1 \cup \cdots \cup T_k$.

We can use this method as a simple algorithm for finding an Euler tour: 
\par\noindent
a) Find a partition of $E$ into edge disjoint cycles.  \par\noindent
b) Iteratively pick a cycle $C$ and merge it with all tours encountered so 
far which have at least one common node with $C$.

Such a merging process certainly converges to a tour covering all nodes, if
a subtour obtained by merging some subtours does not decompose
later into 
some subtours again. If we use a local swapping technique to merge tours, 
this can very well happen, if swapping is again applied to some other node of 
the
merged tour 
(see Figure~\ref{fig:WrongSwap}). In the RAM model we can keep all tours 
in RAM and avoid such fatal nodes.
\begin{figure}
  \centering
  \fbox{
    \includegraphics[width=0.9\linewidth]{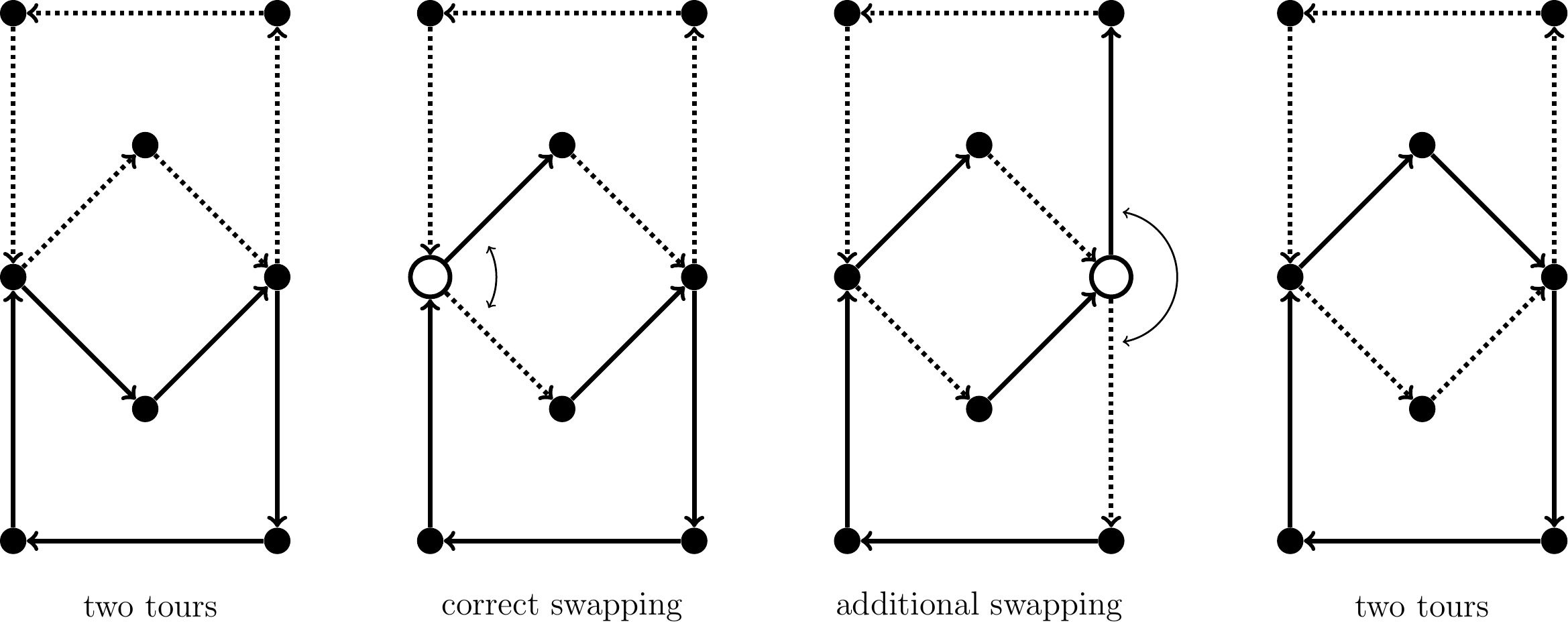}
  }
  \caption{Multiple edge-swapping destroys the merging 
    effect}\label{fig:WrongSwap}
\end{figure}

In the W-stream model with $O(n\log n)$ RAM it is far from being obvious 
how to implement an efficient tour merging for the
following reasons.
\begin{enumerate}
\item We cannot keep every intermediate tour in RAM, so we have to regularly 
  remove some edges together with their successors from RAM, 
  even if we do not know the edges yet to come.
  But on the other hand, we have to keep edges in RAM which are essential 
  in later merging steps.
\item Sometimes we have to merge cycles with tours that had already 
  left RAM. Therefore, we must keep track of common nodes and the related 
  edges.
\end{enumerate}

\subsection{Subtour merging in limited RAM}
Let us assume that we have found say four cycles $C_1,\ldots,C_4$ in that 
order, 
all 
sharing a common node $v$.
(see Figure~\ref{fig:OneNode}).
Let $(u_1,v),\ldots,(u_4,v)$ be the respective in-going edges and 
$(v,w_1),\ldots,(v,w_4)$ be the respective out-going edges.
By swapping the successor edges of $(u_1,v)$ and 
$(u_2,v)$ as explained before, we get a tour containing all edges from $C_1$ 
and $C_2$.
We then merge this tour with $C_3$ swapping the successor edges of 
$(u_1,v)$ and 
$(u_3,v)$, and then with $C_4$ by swapping the successors of $(u_1,v)$ and 
$(u_4,v)$. The successor edges are now as follows:
\begin{align*}
  &(u_1,v) \longrightarrow (v,w_4) & (u_2,v) \longrightarrow (v,w_1)\\
  &(u_3,v) \longrightarrow (v,w_2) & (u_4,v) \longrightarrow (v,w_3)
\end{align*}
For $i > 1$ and cycle $C_i$, the successor of the edge $(u_i,v)$ is 
edge $(v,w_{i-1})$, the out-going edge of $C_{i-1}$.
The edge $(u_1,v)$ of the cycle $C_1$ has the out-going edge of the 
last cycle as its successor edge.
The edge $(u_1,v)$ is the first in-going edge of $v$
called the \emph{first-in edge of $v$}.
\begin{figure}
  \centering
  \fbox{
    \includegraphics[width=0.95\linewidth]{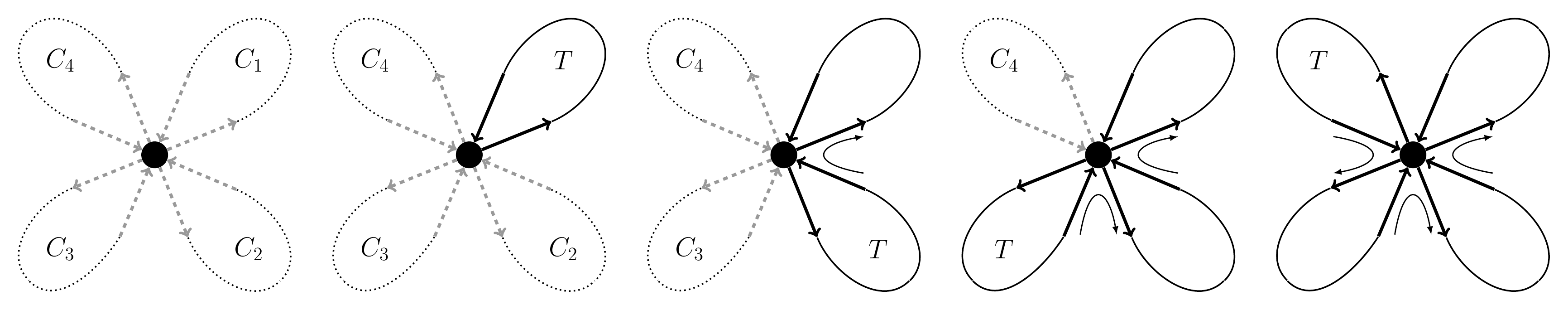}
  }
  % \fbox{
  % \includegraphics[width=0.18\linewidth]{pics/soda_pics/oneNode2.pdf}
  % }
  %   \fbox{
  %   \includegraphics[width=0.18\linewidth]{pics/soda_pics/oneNode3.pdf}
  % }
  %   \fbox{
  %   \includegraphics[width=0.18\linewidth]{pics/soda_pics/oneNode4.pdf}
  % }
  %   \fbox{
  %   \includegraphics[width=0.18\linewidth]{pics/soda_pics/oneNode5.pdf}
  % }
  \caption{Successive merging of cycles}\label{fig:OneNode}
\end{figure}
Let us briefly show how this merging can be implemented in W-streaming.
When $C_1$ is kept in RAM, we store the edge $(u_1,v)$, since we don't know its 
final successor edge yet.
We also keep the edge $(v,w_1)$ in RAM, because it will be the successor edge 
of $C_2$.
We call such an edge the \emph{potential successor edge of $v$}.
We can remove every edge except $(u_1,v)$ together with their respective 
successor 
edges in $C_1$, since only the successor edge of $(u_1,v)$ will change over the 
course of the algorithm.
Then iteratively, if we have a cycle $C_i$ for $i > 1$ in RAM, we assign the 
edge $(v,w_{i-1})$ as successor edge of $(u_i,v)$,
replace $(v,w_{i-1})$ by $(v,w_i)$ in RAM as 
potential successor edge of $v$ for the next cycle and then remove $C_i$ with 
the respective successor edges from RAM.
Finally, as no more cycles with node $v$ occur, we can remove $(u_1,v)$ 
together with 
the last successor edge left from RAM (in our case this is $(v,w_4)$).

Now, let us consider the more complicated case, 
where we wish to merge a cycle $C$ with multiple tours 
at several nodes.
Consider a cycle $C$ and tours $T_1,\ldots,T_j$. Let 
$v_1,\ldots,v_j$, $C$ be nodes so that $v_i$ belongs to $T_i$ and $C$ for all 
$i$.
We distinguish between merging at three types of nodes:
\begin{enumerate}
\item For the nodes $v_1,\ldots,v_j$.
  we use the successor edge 
  swapping.
\item Nodes in $C$ and in $T_1 \cup \cdots \cup T_j\backslash 
\{v_1,\ldots,v_j\}$:
  as only one successor edge swapping per tour is needed, these 
  additional common nodes are not used, so for every $v \in T_1 \cup \cdots 
  \cup T_j\backslash \{v_1,\ldots,v_j\}$ the in-going edge $(u,v)$ 
  of $C$ keeps its successor edge, so nothing happens here.
\item Nodes in $C\backslash (T_1 \cup \cdots \cup T_j)$.
  These nodes are visited by the algorithm for the first time.
  Since we might want to merge $C$ with future cycles at these nodes, we store 
  for every $v \in C\backslash (T_1 \cup \cdots \cup T_j)$ the in-going 
  edge $(u,v)$ of $C$ as first-in edge
  and the out-going edge $(v,w)$ of $C$ as potential successor edge.
  % This is the start of the 
  % successor edge swapping technique at $v$.
\end{enumerate}
Note that the very first cycle found by the algorithm consists only of type 3 
nodes, so every edge will become a first-in edge.

The challenge in the analysis is on the one hand to choose sufficiently many
nodes where merging is done in a {\em simultaneous} way in order to
stay within the one-pass complexity, and on the other hand to ensure that
simultaneous merging enlarges and never decomposes subtours. Here we need two 
lemmas.  Lemma~\ref{technical2} is used to show that merging of equivalence 
classes of Euler subtours
leads to equivalence classes of a new subtour, thus subtour merging is 
invariant 
w.r.t. the equivalence class
relation. This lemma is needed to prove 
Lemma~\ref{mainLemma}, which shows that the sequence of successor 
functions 
iteratively built
by refining the equivalence relation are indeed Eulerian subtours. It also 
gives a criterion for belongingness of edges
to such a subtour. This criterion is finally used to show that the successor
function returned by our algorithm is equal to the successor function
associated to the last and most refined equivalence relation, hence is an 
Eulerian tour for the graph.

%The beauty of our algorithm is that the intermediate tours we build are 
%pairwise node-disjoint all the time.
%Hence, for every node we have to store at most one first-in edge and one 
%current out-going edge (potential successor edge) that can be the successor 
%edge of the next cycle.
%By additionally remembering which intermediate tours are currently occupying 
%which nodes, we can make sure that the merging steps can be proceeded with 
%exactly one common node per tour.
%The necessary information can be stored in $\mathcal O(n~\log(n))$ RAM (see 
%Lemma~\ref{lem:memory}).
%We will prove all of our statements in the following sections.

For the readers convenience we give a high level description of our algorithm.
A detailed description in pseudo code together with an outline of the analysis 
and the proof of the main theorem will follow in 
the next sections.
%Note that at every point of time we store at most one first-in and one 
%potential successor edge per node.
We denote the set of first-in edges by $F$.
  \begin{enumerate}[label=\arabic*.]
  \item Iteratively:
    \begin{enumerate}[label=\arabic*.]
    \item \label{step1} Read edges from the input stream until the edges 
      in RAM contain a 
      cycle $C$.
    \item \label{step2} If a node $v$ of $C$ is visited for the first time,
      \begin{itemize}
      \item[a)] store the in-going edge $(u,v)$ of $C$ in $F$ (we will process 
        these $\leq n$ edges in step~\ref{step7}),
      \item[b)] remember the out-going edge $(v,w)$ as potential successor edge 
      of 
        $v$.
      \end{itemize}
    \item \label{step3}
      % If a node $v$ of $C$ is not visited for the first time, it 
      % already is assigned 
      % to a unique tour $T_i$ with $v \in C \cap T_i$.
      % For each tour, which has at least one common node with $C$, choose 
      % exactly one of them.
      Every node $v$ that has already been visited,
      has thereby been assigned to a unique tour $T$
      with $v\in C\cap T$.
      For each tour that shares a node with $C$, choose 
      exactly one common node.
    \item \label{step4} For each node $v$ chosen in step~\ref{step3}  
``swap 
the 
      successors''.
      That means, we write the in-going edge $e$ to the stream and take the 
      recent potential successor edge of $v$ as successor edge for $e$.
      Then,  save the out-going edge as new potential successor edge of $v$.
    \item \label{step5} For each edge that has not been stored in $F$ 
(step 
      \ref{step2}) or written to the stream (step \ref{step4}) so far, write 
this 
      edge to the stream and take as successor the following edge in $C$.
    \item \label{step6} All tours with common nodes together with 
      all newly 
      visited nodes are now assigned to a single tour.
    \end{enumerate}
  \item \label{step7} After the end of the input stream is reached,
    all edges have either been written to the stream
    or stored in $F$.
    %we process the edges stored in $F$. 
    For every edge $(u,v) \in F$, write it to the stream and take as its 
successor 
    the potential successor edge of $v$.
  \end{enumerate}
 
  An example of how the algorithm works can be found in 
the appendix.

\section{The Algorithm}\label{sec:algorithm}
To enable a clear and structured analysis,
in this section we present the pseudo-code for our algorithm.
For a better understanding
it is split up into several procedures
that correspond to the steps from our high level description in 
Section~\ref{sec:idea}.
Note that these procedures are not independent algorithms,
since they access variables from the main algorithm.
The output is an Euler tour on $G$,
given in the form of a successor function $\delta^*$.
To be more precise, the output is a stream of triples $(v_1,v_2,s)$
with $v_1,v_2,s\in V$ and $\{v_1,v_2\}\in E$.
Each of these triples represents the information $\delta^*((v_1,v_2))=(v_2,s)$.
If a triple $(v_1,v_2,s)$ is written to the stream,
we say that the edge $(v_2,s)$ is \emph{marked as successor} of the edge 
$(v_1,v_2)$.
For every node we store two important values during the algorithm:
The value $t(v)$ that gives the tour $v$ is assigned to at the moment
and the value $j(v)$ that indicates that $(v,j(v))$ is the potential successor 
edge of $v$.
\\
\begin{algorithm}[H]\caption{\textsc{Euler-Tour}}
  \SetKwInOut{Input}{input}
  \SetKwInOut{Output}{output}
  \Indm%
  \Input{Undirected graph $G=(V,E)$, edge by edge on a stream $S$}
  \Output{Euler tour on $G$, i.e.\ a successor function $\delta^*$, if there is 
one}
  \Indp%

  $c:=0$;~
  $F:=\emptyset$;~
  $E_{\mathrm{int}}:=\emptyset$;~
  for every $v\in V$: $j(v):=0,t(v):=0$\;
  
  \For{every edge $e$ on $S$}{
    $E_{\mathrm{int}}:=E_{\mathrm{int}}\cup \{e\}$\;
    \If{$G_{\mathrm{int}}=(V,E_{\mathrm{int}})$ contains a cycle $C$}{
      % $\tilde{C}:=Merge-Cycle (C)$\For this reason, we choose exactly one 
node 
      %of each value for merging.;
      % write $\tilde{C}$ to the stream\;
      % delete $C$ from $E_{\mathrm{int}}$\;
      \textsc{Merge-Cycle} ($C$)\;
    } 
  }
  \If{$E_{\mathrm{int}}=\emptyset$}{
    ERROR:\ At least one node with odd degree\;
  }
  \If{ there exist $u,v$ with $t(u)\neq t(v)\neq 0$}{
    ERROR:\ Graph not connected\;
  }
  \textsc{Write-F}\;
\end{algorithm}
The algorithm searches the stream for cycles
(Step~\ref{step1} in our high level description )
and whenever a cycle is found,
we will run the procedure \textsc{Merge-Cycle} on this cycle.
The procedure \textsc{Merge-Cycle} contains the steps~\ref{step2} to~\ref{step6}
and \textsc{Write-F} corresponds to step~\ref{step7}\\
% This is the essential procedure of our algorithm and is divided into five 
%steps,
% which we will have a closer look at in the following.
% After that we will care about the procedure Write-F.
\begin{procedure}[H]
  \caption{Merge-Cycle()} % chktex 36 
  \SetKwInOut{Input}{input}
  \SetKwInOut{Output}{output}
  \Indm%
  \Input{Ordered cycle $C=(v_1,\ldots, v_k)$ }
  % \Output{Set of directed edges $\tilde{C}$ }
  \Indp%
  \textsc{New-Nodes} \;
  \textsc{Construct-J-M}\;
  \textsc{Merge}\;
  \textsc{Write}\;
  \textsc{Update}\;
  \For{every edge $e\in C$}{
    delete $e$ from $E_{\mathrm{int}}$
  }
\end{procedure}
The procedure \textsc{New-Nodes} implements step~\ref{step2}\ 
If a node $v$ is processed the very first time by the algorithm,
this is indicated by $t(v)=0$.
If this is the case, we store the corresponding in-going edge in the set $F$
and store the next node on the cycle in $j(v)$
(this is, the edge $(v,v_{i+1})$ becomes the potential successor of $v$).\\
\begin{procedure}[H]\caption{New-Nodes()} % chktex 36
  \For{$i=1,\ldots,k$}{
    \If{$t(v_i)=0$}{
      $j(v_i)=v_{i+1}$\;
      $F=F\cup\{(v_{i-1},v_i)\}$\;
      % delete $(v_{i-1},v_i)$ from $C$ and from $E_{\mathrm{int}}$\;
    }
  }
  
\end{procedure}
The procedure \textsc{Construct-J-M} is a realization of step~\ref{step3}\ 
For every value $j\neq 0$, we pick exactly one node $v$
with $t(v)=j$ if there is one.
These nodes are stored in $J$, their values are stored in $M$.
The nodes in $J$ are the ``chosen'' nodes we want to use for merging tours.
If two nodes already have the same value in $t$,
this means they are already part of the same tour (see Lemma~\ref{mainLemma})
and we want to avoid using both of them for merging.\\
\begin{procedure}[H]\caption{Construct-J-M()} % chktex 36
  $M=\emptyset$;~
  $J=\emptyset$\;
  \For{$j=1,\ldots,|V|$}{
    \If{exists $i\in[k]$ with $t(v_i)=j$}{
      add exactly one $v_i$ with $t(v_i)=j$ to the set $J$\;
      $M=M\cup\{j\}$\;
    }
  }
\end{procedure}
In the following procedure \textsc{Merge},
we use the nodes from $J$ to merge all tours
that share a node in the cycle $C$ by edge-swapping (step~\ref{step4}).\\
\begin{procedure}[H]\caption{Merge()} % chktex 36
  \For{each $v_i\in J$}{
    % mark $(v_i,j(v_i))$ as successor for $(v_{i-1},v_i)$\;
    write $(v_{i-1},v_i,j(v_i))$ to the stream\;
    $j(v_i)=v_{i+1}$\;
  }
\end{procedure}
In the procedure \textsc{Write}, we take care of all the edges
that have not been stored in $F$
and have not been written to the stream in the procedure \textsc{Merge} 
(Step~\ref{step5}).\\
\begin{procedure}[H]\caption{Write()} % chktex 36
  \For{each edge $(v_i,v_{i+1})\in C$ that has not been written to the stream 
or 
added to $F$}{
    % mark $(v_{i+1},v_{i+2})$ as successor for $(v_i,v_{i+1})$\;
    write $(v_i,v_{i+1},v_{i+2})$ to the stream\;
  }
\end{procedure}
In the procedure \textsc{Update}
we update the $t$-values to implement step~\ref{step6}\ 
After this step we can be sure that any two nodes
$v,v'\in V$ with $t(v)=t(v')\neq 0$
belong to the same tour,
whereas $t(v)=0$ means that $v$ has not been processed so far.\\
\begin{procedure}[H]\caption{Update()} % chktex 36
  $a:=0$\;
  \If{$M=\emptyset$}{
    $c:=c+1$\;
    $a:=c$\;
  }
  \Else{
    $a:=\min(M)$\;
  }
  \For{each $v\in V$}{
    \If{$t(v)\in M$}{
      $t(v):=a$\;
    }
  }
  \For{$i=1,\ldots,k$}{
    $t(v_i)=s$\;
  }
\end{procedure}

Finally, in the procedure \textsc{Write-F} (step~\ref{step7}), the first-in 
edges
that have been stored in $F$ during the algorithm
are written to the stream with proper successors.\\
\begin{procedure}[H]\caption{Write-F()} % chktex 36
  \For{each edge $(u,v)\in F$}{
    write $(u,v,j(v))$ to the stream\;
  }
\end{procedure}

\section{Analysis}\label{sec:analysis}
\subsection{Subtour representation by equivalence classes}
%\subsection{Transferring the Problem}
In this subsection we present some basic definitions and results
that allow us to transfer the problem of tour merging in a graph
to the language of equivalence relations.
This will allow an elegant and clear analysis of our algorithm in 
Section~\ref{sec:analysis}.
\begin{definition}
  \begin{itemize}
  \item [(i)] Let $G=(V,E)$ be an undirected graph.
    An \emph{orientation} of the edges of $G$ is a function $R:E\rightarrow V^2$
    such that for every edge $\{u,v\}\in E$ either $R(\{u,v\})=(u,v)$ or 
$R(\{u,v\})=(v,u)$.
    So $R(G):=(V,R(E))$ is a directed graph.
  \item[(ii)] Let $\vec{G}=(V,\vE)$ be a directed graph.
    A \emph{successor function} on $\vG$ is a function
    $\delta:\vE\rightarrow \vE$ with ${\delta(e)}_{(1)}=e_{(2)}$ for all 
$e\in\vE$.
  \item[(iii)] Let $\vG=(V,\vE)$ be a directed graph with successor function 
$\delta$.
    We define the relation $\equiv_\delta$ on $\vE$ by
    $e\eqd e':\Leftrightarrow \exists k\in\mathbb{N}:\delta^k(e)=e'$,
    where $\delta^k$ denotes the $k$-wise composition of $\delta$.    
  \end{itemize}
\end{definition}
So $e\eqd e'$ means that $e'$ can be reached from $e$ by iteratively applying 
$\delta$.

\begin{lemma}\label{lem:eqrel}
  Let $\delta$ be a bijective successor function on a directed graph 
$\vG=(V,\vE)$.
  Then $\eqd$ is an equivalence relation on $\vE$. 
\end{lemma}

\begin{proof}
  Reflexivity: Let $e\in \vE$. Since $\vE$ is finite, there exists $k\in\N$ 
with 
the following property:
  There exists $k'\in\N$  with $k'<k$ and $\delta^k(e)=\delta^{k'}(e)$.
  Let $k$ be minimal with this property.
  Since $\delta$ is injective,
  it follows that $\delta^{k-1}(e)=\delta^{k'-1}(e)$
  and the minimality of $k$ enforces that $k'-1\notin\N$.
  So $k'=1$, therefore
  $\delta^k(e)=\delta(e)$ and by injectivity of $\delta$ we have 
$\delta^{k-1}(e)=e$.\par
  Symmetry: Let $e,e'\in\vE$ with $e\eqd e'$.
  Then there exists a minimal $k\in\N$ with $\delta^k(e)=e'$.
  As shown above, there also exists a $k'\in\N$ with $\delta^{k'}(e)=e$.
  Because $k$ is minimal, we have $k<k'$.
  It follows that $\delta^{k'-k}(e')=\delta^{k'}(e)=e$.\par
  Transitivity: Let $e,e',e''\in\vE$ with $e\eqd e'$ and $e'\eqd e''$.
  Then there exist $k_1,k_2\in\N$ with $\delta^{k_1}(e)=e'$
  and $\delta^{k_2}(e')=e''$. So we have $\delta^{k_1+k_2}(e)=e''$.  
\end{proof}

We denote the equivalence class of an edge $e\in\vE$ w.r.t. $\eqd$ by 
${[e]}_\delta$.

The following lemma is necessary to show that the equivalence classes of 
$\delta$
always form tours on $\vG$.
\begin{lemma}\label{lem:technical1}
  Let $\vG=(V,\vE)$ be a directed graph with bijective successor function 
$\delta$
  and the related equivalence relation $\eqd$.
  Then we have:
  \begin{itemize}
  \item[(i)]
    Let $e\in \vE$ and $k_1,k_2\in\N_0$ with $k_1\neq k_2$
    and $\delta^{k_1}(e)=\delta^{k_2}(e)$.
    Then $|k_1-k_2|\geq|{[e]}_\delta|$.
  \item[(ii)]
    For any $e\in\vE$ we have $\delta^{|{[e]}_\delta|}(e)=e$. 
    % \item[(iii)]
    %   Let $e,e'\in\vE$ with $e\neqd e'$
    %   and let $\delta'$ be a successor function on $\vG$ with
    %   $\delta'|_{\vE\setminus\{e,e'\}}=\delta|_{\vE\setminus\{e,e'\}}$ and
    %   $\delta'(e)=\delta(e')$ and $\delta'(e')=\delta(e)$.
    %   Then ${[e]}_{\delta'}={[e']}_{\delta'}={[e]}_\delta\cup{[e']}_\delta$
    %   and for any $e''\in\vE$ with $e''\neqd e,e''\neqd e'$ it holds
    %   $[e'']_{\delta}=[e'']_{\delta'}$.
  \end{itemize}
\end{lemma}

\begin{proof}
  (i):
  Assume for a moment that there exist $e\in\vE$ and $k_1,k_2\in\N$
  with $\delta^{k_1}(e)=\delta^{k_2}(e)$ and $0<|k_1-k_2|<|{[e]}_\delta|$.
  Without loss of generality let $k_1>k_2$.
  We have $\delta^{k_1-k_2}(\delta^{k_2}(e))=\delta^{k_1}(e)=\delta^{k_2}(e)$
  and via induction for every $s\in\N$, we get
  $\delta^{s(k_1-k_2)}(\delta^{k_2}(e))$ $=\delta^{k_2}(e)$.
  For the set $M:=\{\delta^k(e)|k_2\leq k<k_1\}$,
  we have $|M|\leq k_1-k_2<|{[e]}_\delta|$.
  But on the other hand, we also have ${[e]}_\delta\subseteq M$:
  Let $e'\in{[e]}_\delta={[\delta^{k_2}(e)]}_\delta$. Let $n\in\N$ with 
$e'=\delta^{n}(\delta^{k_2}(e))$.
  Then there exist unique $s,r\in\N_0$ with $0\leq r< k_1-k_2$ and 
$n=s(k_1-k_2)+r$.
  So
  
\[e'=\delta^n(\delta^{k_2}(e))=\delta^{r}(\delta^{s(k_1-k_2)}(\delta^{k_2}
(e)))=\delta^r(\delta^{k_2}(e))=\delta^{k_2+r}(e)\in M.\]
  Now we have 
  $|M| \leq k_1 - k_2 < |{[e]}_{\delta}| \leq |M|$,
  a contradiction.
  \par(ii):
  Assume that there exists $e\in\vE$ with $\delta^{|{[e]}_\delta|}(e)=e'\neq e$.
  Define $M:=\{\delta^k(e)|1\leq k\leq|{[e]}_\delta|\}$.
  Clearly, $M\subset {[e]}_{\delta}$.\par
  Case 1: $e\in M$. Then $\delta^0(e)=e=\delta^k(e)$ for some $k$
  with $1\leq k<|{[e]}_\delta|$, in contradiction to (i).\par
  Case 2: $e\notin M$. Then $|M|<|{[e]}_\delta|$,
  By pigeon hole principle, there exist $1\leq k_1,k_2\leq|{[e]}_\delta|$
  with $\delta^{k_1}(e)=\delta^{k_2}(e)$
  in contradiction to (i).\par
\end{proof}

\begin{theorem}\label{thm:classifyEulerTour}
  Let $\vG=(V,\vE)$ be a directed graph with bijective successor function 
$\delta$
  such that $e\eqd e'$ for all $e,e'\in\vE$.
  Then $\delta$ determines an Euler tour on $\vG$ in the following sense:
  For every $e\in \vE$ the sequence 
$(e_{(1)},{\delta(e)}_{(1)},\ldots,\delta^{|\vE|}{(e)}_{(1)})$ is an Euler tour 
on $\vG$.
\end{theorem}
\begin{proof}
  Let $e\in\vE$. Note that ${[e]}_\delta=\vE$.
  The sequence 
  $(e_{(1)},{\delta(e)}_{(1)},\ldots,\delta^{|\vE|}{(e)}_{(1)})$
  consists of $|\vE|$ edges, namely
  $e,\delta(e),\ldots,\delta^{|\vE|-1}(e)$.
  These edges are pairwise distinct:
  Otherwise, we would have $\delta^{k_1}(e)=\delta^{k_2}(e)$
  for some $k_1,k_2\in\{0,\ldots,|\vE|-1\}$.
  Hence, $|k_1-k_2|<|\vE|$
  in contradiction to Lemma~\ref{lem:technical1} (i).
  So the sequence is a trail.
  By applying Lemma~\ref{lem:technical1} (ii), we get
  $e=\delta^{|{[e]}_\delta|}(e)=\delta^{|\vE|}(e)$,
  thus the trail is a tour on $\vG$ and since it has length $|\vE|$,
  it is an Euler tour on $\vG$.
\end{proof}

Before we start with a detailed memory- and correctness analysis,
we show that at the end of the algorithm, every edge  $\{u,v\}\in E$
has been written to the output stream exactly once,
either in the form $(u,v)$ or in the form $(v,u)$.
We also show that $|E_{\mathrm{int}}|\leq n$ all the time.

\begin{lemma}\label{lem:everyEdgeWritten}
  \begin{itemize}
  \item[(i)] After each processing of an edge (lines $2$ to $5$ in \textsc{ 
Euler-Tour})
    in the algorithm, the graph $G_{\mathrm{int}}=(V,E_{\mathrm{int}})$ is 
cycle-free so
    $|E_{\mathrm{int}}|\leq n$.  If all nodes from $V$ have even degree in $G$, 
after
    completion of \textsc{ Euler-Tour}, $E_{\mathrm{int}}=\emptyset$.
  \item[(ii)] If all nodes from $V$ have even degree in $G$,
    after completion of \textsc{ Euler-Tour} every edge $\{u,v\}\in E$
    has been written to the stream either in the form $(u,v,s)$
    or in the form $(v,u,s)$ for some $s\in V$.
  \end{itemize}
\end{lemma}

\begin{proof}
  We start by proving the first part of (i) via induction
  over the number of already processed edges.
  If there are no edges processed so far, then $E_{\mathrm{int}}=\emptyset$,
  so $G_{\mathrm{int}}$ is cycle-free.
  Now let $k\in[|E|]\cup\{0\}$, let $G_k,G_{k+1}$ denote $G_{\mathrm{int}}$
  after $k$ resp.\ $k+1$ edges have been processed
  and let $G_k$ be cycle-free.
  Let $e$ denote the $(k+1)$-th processed edge.
  When $e$ is added to $G_{\mathrm{int}}$, it may produce a cycle $C$.
  If $e$ does not produce a cycle, then $G_{k+1}=G_k\cup\{e\}$ is cycle-free
  and we are done.
  If $e$ produces a cycle $C$,
  then (at lines $6,7$ in \textsc{Merge-Cycle})  $C$ is deleted from 
$E_{\mathrm{int}}$
  and because $e\in C$, we get $G_{k+1}=(G_k\cup\{e\})\setminus C\subseteq G_k$
  and we are done by the induction hypothesis.\par
  Now assume for a moment that $E_{\mathrm{int}}\neq \emptyset$
  at the end of \textsc{ Euler-Tour}.
  We know that $G_{\mathrm{int}}$ is cycle-free at this time,
  so $G_{\mathrm{int}}$ contains a node with odd degree in $G_{\mathrm{int}}$.
  Because we always delete whole cycles,
  the degree of this node in $G$ has to be odd as well,
  but then $G$ is not an Eulerian graph.
  In this case we might output a message that $G$ does not contain an Euler 
tour.\par
  About (ii). During \textsc{ Euler-Tour}, every edge from $E$ is added to 
$E_{\mathrm{int}}$
  at some point of time
  and there is only one way for an edge to be deleted from $E_{\mathrm{int}}$ 
again,
  namely in line $7$ of \textsc{Merge-Cycle}.
  At that point of time, the edge has either been written to the stream
  in \textsc{Merge} or \textsc{Write} (in which case we are done)
  or it has been added to $F$ in \textsc{New-Nodes}.
  In that case it is written to the stream in \textsc{Write-F}.
  Because, according to (i), $E_{\mathrm{int}}=\emptyset$ at the end of 
\textsc{ 
Euler-Tour},
  at this point of time, every edge must have been written to the stream
  in exactly one of the two ways.
\end{proof}

The idea of (i) is that every
time a cycle occurs in $E_{\mathrm{int}}$, we delete this cycle so we assure 
that
$E_{\mathrm{int}}$ becomes cycle-free again (since we only add one edge at a
time).

\subsection{Memory Requirement}

For the memory estimation, we have to consider
the variables $j(v),t(v)$ for all $v\in V$,
the sets $F,E_{\mathrm{int}},J$, and $M$,
and the counter $c$.
By Lemma~\ref{lem:everyEdgeWritten} (i), $|E_{\mathrm{int}}|\leq n$
and, with some straightforward considerations,
we can estimate the memory requirement for the other parameters
leading to the following lemma.

\begin{lemma}\label{lem:memory}
  Algorithm~\textsc{Euler-Tour} needs at most $\mathcal{O}(n\log n)$ bits of 
RAM\@.
\end{lemma}

\begin{proof}
  We consider the different parameters.
  \par
  About $c$: We show that $c\leq n/3$ at every time,
  so $\log n$ bits suffice to store $c$.
  $c$ is initiated with $0$
  and changed in the procedure \textsc{Update}
  if and only if $M=\emptyset$ at that point of time.
  This only happens
  if for every node $v$ of the considered cycle, we have $t(v)=0$,
  which means that none of the cycle nodes
  was considered before.
  This case can occur at most $n/3$ times
  during the algorithm, because there can be no more than $n/3$ node disjoint 
cycles in $G$, so $c\leq n/3$.  
  \par
  About $j(v)$: In this variable we store the label of a node,
  so for fixed $v$, $\log n$ bits suffice
  and altogether $n\log n$ bits suffice.
  \par
  About $t(v)$: We prove that for any $v\in V$
  $t(v)\leq n$ at any time:
  Assume for a moment that this is not the case.
  Consider the first point of time $T$
  in which $t(v)$ is set to a value $>n$ for some $v\in V$.
  $t(v)$ is only changed in the procedure \textsc{Update},
  line $9$ or $11$.
  In both cases the value is set to $r$
  which is either $c$ (line $4$) or $\min(M)$ (line $6$).
  We already showed $c<n$.
  Hence, by our assumption,
  $\min(M)>n$ at that point of time.
  But this implies that at the time of the construction of $M$,
  there already existed a node $u\in V$ with $t(u)>n$
  in contradiction to the choice of $T$.
  \par
  About $E_{\mathrm{int}},F,J,M$:
  Because a single element of each of these sets
  can be stored in $\log n$ bits,
  it suffices to show that the cardinalities of these sets
  do not exceed $n$.
  For $E_{\mathrm{int}}$, this is shown in Lemma~\ref{lem:everyEdgeWritten}.
  For $J$ and $M$, it follows directly from the construction
  (see Procedure \textsc{Construct-J-M}).
  In the set $F$, for every node we collect the first edge
  that leads into this node (see Procedure \textsc{New-Nodes}, lines $2$ and 
$4$),
  so clearly $|F|\leq n$.
\end{proof}

\subsection{Correctness}

In this subsection, we prove that $\delta^*$ determines an Euler tour on $G$,
provided that $G$ is Eulerian (Theorem~\ref{finalthm}).
This is done with the help of Theorem~\ref{thm:classifyEulerTour},
where bijectivity of $\delta^*$
and the condition that $\delta^*$ induces only one equivalence class
is required.
In the following, we show that these assumptions are true for $\delta^*$
by generating a sequence of bijective successor functions
$\delta^*_0,\ldots,\delta^*_N$ such that $\delta^*_0$ is bijective,
$\delta^*_N=\delta^*$
and $\delta_{i+1}^*$ emerges from $\delta_{i}^*$ by swapping
of edge successors.

Lemma~\ref{lem:everyEdgeWritten} (ii) induces an orientation on $E$ which we 
call $R^*$:
For all $\{u,v\}\in E$, we define
\[R^*(\{u,v\}):=(u,v) \text{ if }(u,v)\text{ has been written to the output 
stream}.\]

From now on, let $C_1,\ldots, C_N$ denote the cycles found in 
$E_{\mathrm{int}}$ 
by the algorithm
in chronological order.
For $k \in \{0,\ldots,N\}$ and a variable $x$, we denote by $x_k$ the value of 
$x$ after the $k$-th call of \textsc{Merge-Cycle}.
For $k = 0$, this means the initial value of $x$.

\begin{definition}
  For each $i\in[N]$, let $C_i=(v_1^{(i)},\ldots,v_{\ell_i}^{(i)})$
  be the form of the cycle given to \textsc{Merge-Cycle}.
  Define $\delta^c_i:E(C_i)\rightarrow E(C_i)$ by
  \mbox{$\delta^c_i(v_j^{(i)},v_{j+1}^{(i)}):=(v_{j+1}^{(i)},v_{j+2}^{(i)})$} 
for every $j\in[k_i]$
  and let $\delta^c:R^*(E)\rightarrow R^*(E)$ denote the unique successor 
function
  with $\delta^c|_{E(C_i)}=\delta^c_i$ for all $i\in[N]$.
\end{definition}
So $\delta^c$ is the natural successor function
induced by the cycles $C_1,\ldots,C_N$. 
\begin{lemma}\label{lem:delc}
  The successor function $\delta^c$ is bijective
  and for any two edges $e,e'$ we have
  $e\equiv_{\delta^c}e'\Leftrightarrow \exists i\in[N]: e,e'\in C_i$.
\end{lemma}

\begin{proof}
  We first show that $\delta^c$ is surjective:
  Let $e\in R^*(E)$. Then there exist $k\in[N]$ and $i\in\N$
  such that $C_k=(v_1,\ldots, v_{\ell_k})$ and $e=(v_i,v_{i+1})$.
  Then $\delta(v_{i-1},v_i)=(v_i,v_{i+1})=e$.
  Because $R^*(E)$ is finite, $\delta^c$ is bijective.\par
  Now let $e,e'\in R^*(E)$ with $e\eqc e'$.
  Let $i\in[N]$ such that $e\in C_i$.
  Since $\delta^c(C_i)=C_i$, it follows that $e'\in C_i$.
  \par
  Now let $e,e'\in C_i$ for some $i\in[N]$,
  for instance $C_i=(v_1,\ldots,v_{\ell_i})$
  and $j,k\in[\ell_i]$ with $e=(v_j,v_{j+1})$ and $e'=(v_k,v_{k+1})$.
  W.l.o.g.\ let $j<k$ and set $r:=k-j$.
  Then ${(\delta^c)}^r(e)=e'$, so $e\equiv_{\delta^c}e'$.
\end{proof}

Let now $k\in\{0,\ldots N\}$. We consider the point of time right
after the $k$-th iteration of \textsc{Merge-Cycle}% and right before the next 
%edge is processed 
(for $k=0$ this means the very beginning of the algorithm).
We call edges from $\bigcup\limits_{i=1}^k E(C_i)$ \emph{processed edges},
since those edges have already been loaded into and then deleted from 
$E_{\mathrm{int}}$.
All processed edges can be divided into two types:
\begin{itemize}
\item Type A:\ The edge has been written on the stream with a dedicated 
successor.
\item Type B:\ The edge has been added to $F$.
\end{itemize}
These are the only possible cases for processed edges
because an edge is deleted from $E_{\mathrm{int}}$
is either written to the stream or added to $F$.
This leads to the following definition.
\begin{definition}
  For every $k\in\{0\ldots,N\}$ define the function
  $\delta_k:\bigcup\limits_{i=1}^k E(C_i)\rightarrow \bigcup\limits_{i=1}^k 
E(C_i)$ by 
  \[(u,v)\mapsto \begin{cases}
      e'      &\text{if } (u,v) \text{ is of type A with successor }e'\\
      (v,j_k(v))&\text{if } (u,v) \text{ is of type B}
    \end{cases}\]
  % and define
  \[\text{ and define }\delta^*_k:=\begin{cases}
      \delta_k \text{ on }\bigcup\limits_{i=1}^{k}E(C_i)\\
      \delta^c \text{ on }\bigcup\limits_{i=k+1}^{N}E(C_i).
    \end{cases}\]
\end{definition}
Note that $\delta^*_0=\delta^c$ and $\delta^*_N=\delta^*$.

\begin{lemma}\label{obs:obs}
  Let $k,\ell\in\{0,\ldots,N\}$ with $k<\ell$.
  Then for any $v,v'\in V$, $e\in R^*(E)$, we have
  \begin{itemize}
  \item[(i)] If $t_k(v)=t_k(v')\neq 0$, then $t_{\ell}(v)=t_{\ell}(v')$.
  \item[(ii)] If $e\in C_{\ell}$, then ${[e]}_{\delk}={[e]}_{\delta^c}$.
  \end{itemize}
\end{lemma}

\begin{proof}
  About (i). Let $v,v'\in V$ with $t_k(v)=t_k(v')\neq 0$.
  Assume for a moment that $t_{\ell}(v)\neq t_{\ell}(v')$.
  Then there exists $k\leq k'< \ell$ such that
  $t_{k'}(v)=t_{k'}(v')$ and $t_{k'+1}(v)\neq t_{k'+1}(v')$.
  Furthermore, $t_{k'}(v)\neq 0$, because $t_k(v)\neq 0$
  and the value $t(v)$ is never set to $0$ after its initiation.
  We take a closer look at the $(k'+1)$-th call of \textsc{Merge-Cycle}.
  If for a node its $t$-value is changed in this call,
  it is set to $a_{k'+1}$ (line $9$ or $11$ in \textsc{Update}),
  so we may assume that $t_{k'+1}(v)=a_{k'+1}\neq t_{k'+1}(v')$.
  But this implies that $t_{k'}(v)\in M$ or $v\in C_{k'+1}$,
  in which case also $t_{k'}(v)\in M$ (since $t_{k'}(v)\neq 0$). 
  But then, $t_{k'}(v')=t_{k'}(v)\in M$ and therefore 
$t_{k'+1}(v')=a_{k'+1}=t_{k'+1}(v)$,
  in contradiction to our assumption.\par
  About (ii). Let $e\in C_\ell$.
  With Lemma~\ref{lem:delc}, we get ${[e]}_{\delk}=E(C_\ell)$.
  Since $\ell>k$, $\delk (e')=\delta^c (e')$ for any $e'\in C_\ell$.
  Hence, $\delk (e)=\delta^c (e)$
  and using $\delta^c(e)\in C_\ell$,
  by induction ${(\delk)}^j(e)={(\delta^c)}^j(e)$ for any $j>1$,
  which proves the claim.
\end{proof}

The following two lemmata form the technical foundation of our analysis.
In Lemma~\ref{technical2} we repeat in a formal way the basic idea of
tour-merging given in Section~\ref{sec:idea}.
It is needed for the proof of Lemma~\ref{mainLemma}.

\begin{lemma}\label{technical2}
  Let $\vG=(V,\vE)$ be a directed graph with bijective successor function 
$\delta$
  and the related equivalence relation $\eqd$.
  Let $r\in\N$ and $e_1,\ldots,e_r\in\vE$
  with $e_i\eqd e_j$ for every $i,j\in[r]$.
  Let $e_1',\ldots,e_r'\in\vE$
  with $e_i'\neqd e_j'$ and $e_i\neqd e_i'$ for every $i,j\in[r]$.
  Let $\delta'$ be a successor function on $\vG$ with
  $\delta'(e)=\delta(e)$ for every
  $e\in \vE\setminus\{e_1,\ldots,e_r,e_1',\ldots,e_r'\}$
  and $\delta'(e_i)=\delta(e_i')$ and $\delta'(e_i')=\delta(e_i)$
  for any $i\in[r]$.
  Then, $\delta'$ is bijective and
  \begin{align}    
&{[e_1]}_{\delta'}=\bigcup\limits_{i=1}^{r}{[e_i']}_\delta\cup{[e_1]}_{\delta}
\tag{P1}\label{prop1}\\
    % \end{equation}
    % \begin{equation}
    &{[e]}_{\delta'}={[e]}_\delta \text{ for any } 
e\in\vE\setminus{[e_1]}_{\delta'}.\tag{P2}\label{prop2}
  \end{align}
\end{lemma}

\begin{proof}
  Via induction over $r$.
  First of all notice that because of the definition of $\delta'$
  and because $\delta$ is bijective, $\delta'$ is bijective as well.
  For $r=1$ to shorten notation, we write $e$ and $e'$ instead of $e_1$ and 
$e_1'$.
  We first show
  \begin{equation}\label{techeq1}
    {[e]}_{\delta'}\subseteq {[e]}_\delta\cup{[e']}_\delta:
  \end{equation}
  First we show that for any $e''\in {[e]}_\delta\cup{[e']}_\delta$,
  we have $\delta'(e'')\in {[e]}_\delta\cup{[e']}_\delta$:
  Let $e''\in{[e]}_\delta\cup{[e']}_\delta$.
  Then there exists $k\in\N$
  such that $e''=\delta^k(e)$ or $e''=\delta^k(e')$.
  If $e''\in\{e,e'\}$, then $\delta'(e'')\in\{\delta(e),\delta(e')\}$
  and otherwise $\delta'(e'')=\delta(e'')=\delta^{k+1}(e)$
  or $\delta'(e'')=\delta^{k+1}(e')$, respectively.
  So in each case 
  we have $\delta'(e'')\in {[e]}_\delta\cup{[e']}_{\delta}$.
  % \textbf{Einschub 1}\par
  % It follows directly from the definition of $\delta'$
  % that for every $e''\in{[e]}_\delta\cup{[e']}_{\delta}$
  % we have $\delta'(e'')\in {[e]}_\delta\cup{[e']}_{\delta}$ as well.
  Since $e\in{[e]}_\delta\cup{[e']}_{\delta}$,
  it follows by induction on $n$ that 
${(\delta')}^n(e)\in{[e]}_\delta\cup{[e']}_{\delta}$
  for any $n\in\N$, so ${[e]}_{\delta'} \subseteq {[e]}_{\delta} \cup 
{[e']}_{\delta}$. \par
  Next, we show
  \begin{equation}\label{techeq2}
    {[e]}_\delta\cup{[e']}_\delta\subseteq{[e']}_{\delta'}.
  \end{equation}
  
  Let $e''\in{[e']}_\delta$.
  Then there exists $k\in\{1,\ldots,|{[e']}_\delta|\}$
  with $e''=\delta^k(e')$.
  Since $e\notin{[e']}_\delta$
  and $\delta^\ell(e')\neq e'$ for all $\ell\in\{1,\ldots,k-1\}$
  (follows from Lemma~\ref{lem:technical1} (i)),
  we have
  \[\delta^{k}(e')=\delta(\delta^{k-1}(e'))=\delta'(\delta^{k-1}(e'))
    =\delta'(\delta'(\delta^{k-2}(e')))=\cdots ={(\delta')}^{k-1}(\delta(e')).\]
  Hence 
$e''=\delta^k(e')={(\delta')}^{k-1}(\delta(e'))={(\delta')}^{k-1}(\delta'(e))={
(\delta')}^k(e)\in{[e]}_{\delta'}$.
  So we have
  \begin{equation}\label{techeq3}
    {[e']}_\delta\subseteq{[e]}_{\delta'}
  \end{equation}
  
  and analogously we get
  \begin{equation}\label{techeq4}
    {[e]}_\delta\subseteq{[e']}_{\delta'},
  \end{equation}
  % so it remains to show ${[e]}_{\delta'}={[e']}_{\delta'}$.
  Because $\delta(e')\in{[e']}_{\delta}\subseteq{[e]}_{\delta'}$,
  we have ${[\delta(e')]}_{\delta'}={[e]}_{\delta'}$
  and it follows that
  \begin{equation}\label{techeq5}    
{[e]}_{\delta'}={[\delta(e')]}_{\delta'}={[\delta'(e)]}_{\delta'}={[e']}_{
\delta'}.
  \end{equation}
  Combining~\eqref{techeq3},~\eqref{techeq4}, and~\eqref{techeq5},
  we proved~\eqref{techeq2}.
  With~\eqref{techeq1},~\eqref{techeq2}, and~\eqref{techeq5}, we have
  
\[{[e]}_{\delta'}\subseteq{[e]}_\delta\cup{[e']}_\delta\subseteq{[e']}_{\delta'}
={[e]}_{\delta'},\]
  so property~\eqref{prop1} is proven.
  For~\eqref{prop2}, note that
  $\delta^k(e'')={(\delta')}^k(e'')$ for any $e''$ with $e''\neqd e, e''\neqd 
e'$
  and any $k\in\N$.
  \par
  \emph{Induction step:}
  Now let $r\in \N$ and let the claim be true for all $k\leq r\in\N$.
  Let $e_1,\ldots,e_{r+1}\in \vE$ with $e_i\eqd e_j$ for every $i,j\in[r+1]$.
  Let $e_1',\ldots, e_{r+1}'\in\vE$ with $e_i'\neqd e_j'$ and $e_i'\neqd e_i$
  for every $i\neq j\in[r+1]$. Let $\delta'$ be a successor function on $\vG$
  with $\delta'(e)=\delta(e)$ for every 
$e\in\vE\setminus\{e_1,\ldots,e_{r+1},e_1',\ldots e_{r+1}'\}$
  and $\delta'(e_i)=\delta(e_i')$ and $\delta'(e_i')=\delta(e_i)$ for every 
$i\in[r+1]$.
  We define a successor function $\delta_r$ for $\vG$ by
  \[\delta_r:=\begin{cases}
      \delta' \text{ on } \vE\setminus\{e_{r+1},e'_{r+1}\}\\
      \delta \text{ on } \{e_{r+1},e'_{r+1}\}.
    \end{cases}\]
  
  With the induction hypothesis applied to $\delta$ and $\delta_r$, we get 
by~\eqref{prop1}
  \begin{equation}\label{deldelm1}    
{[e_1]}_{\delta_r}=\bigcup\limits_{i=1}^{r}{[e_i']}_{\delta}\cup{[e_1]}_\delta
  \end{equation}
  
  and by~\eqref{prop2}
  \begin{equation}\label{deldelmp2}
    {[e_{r+1}']}_{\delta_r}={[e_{r+1}']}_{\delta}.
  \end{equation}
  % We apply the induction hypothesis again,
  % this time for $\delta'$ and $\delta_m$ with $e_{m+1},e_{m+1}'$ and get
  Now we apply the induction hypothesis to $\delta_r$ and $\delta'$ as follows:
  We take $\delta_r$ instead of $\delta$, $\delta'$ remains, $r=1$,
  $e_1$ resp. $e_1'$ are replaced by $e_{r+1}$ resp. $e_{r+1}'$.
  This gives
  \begin{equation}\label{deldelm2}
    {[e_{r+1}]}_{\delta'}={[e_{r+1}']}_{\delta_r}\cup{[e_{r+1}]}_{\delta_r}.
  \end{equation}
  Since $e_1\eqd e_{r+1}$, we get with~\eqref{deldelm1}
  \[e_{r+1}\in{[e_{r+1}]}_{\delta}={[e_1]}_\delta\subseteq{[e_1]}_{\delta_r}\]
  which implies 
  \begin{equation}\label{deldelm2a}
    {[e_{r+1}]}_{\delta_r}={[e_1]}_{\delta_r}.
  \end{equation}
  Summarizing, we have 
  \begin{align}
    {[e_{r+1}]}_{\delta'}  
&\stackrel{\eqref{deldelm2}}{=}{[e_{r+1}']}_{\delta_r}\cup{[e_{r+1}]}_{\delta_r}
\notag \\   
&\stackrel{\eqref{deldelm2a}}{=}{[e_{r+1}']}_{\delta_r}\cup{[e_1]}_{\delta_m}
\label{eq1}\\
    &\stackrel{\eqref{deldelm1}}{=}{[e_{r+1}']}_{\delta_r}\cup\Big( 
\bigcup\limits_{i=1}^{r}{[e_i']}_{\delta}\cup{[e_1]}_\delta\Big) \notag\\
    &\stackrel{\eqref{deldelmp2}}{=}{[e_{r+1}']}_{\delta}\cup\Big( 
\bigcup\limits_{i=1}^{r}{[e_i']}_{\delta}\cup{[e_1]}_\delta\Big) \notag\\
    &=\bigcup\limits_{i=1}^{r+1}{[e_i']}_{\delta}\cup{[e_1]}_\delta.\label{eq2}
  \end{align}
  So~\eqref{prop1} is proved, if ${[e_{r+1}]}_{\delta'}={[e_1]}_{\delta'}$.
  By~\eqref{eq2} ${[e_1]}_\delta\subseteq{[e_{r+1}]}_{\delta'}$,
  so $e_1\in{[e_{r+1}]}_{\delta'}$ and hence 
  \begin{equation}\label{deldelm2b}
    {[e_{r+1}]}_{\delta'}={[e_1]}_{\delta'}. 
  \end{equation}
  \par
  
  For the proof of~\eqref{prop2}, let $e\in\vE\setminus{[e_1]}_{\delta'}$.
  Since $e\notin{[e_1]}_{\delta'}$, by~\eqref{eq1} and~\eqref{deldelm2b} 
$e\notin{[e_1]}_{\delta_r}$.
  Applying~\eqref{prop2} of the induction hypothesis to $\delta$ and $\delta_r$,
  gives us ${[e]}_{\delta_r}={[e]}_{\delta}$.
  We know ${[e_{r+1}]}_{\delta'}={[e_1]}_{\delta'}$,
  so $e\notin{[e_{r+1}]}_{\delta'}$.
  As above, we apply the induction hypothesis to $\delta_r$ and $\delta'$
  and get ${[e]}_{\delta'}={[e]}_{\delta_r}$.
  Altogether ${[e]}_{\delta'}={[e]}_{\delta_r}=[e_\delta]$.
\end{proof}

Note that $\delta'$ emerges from $\delta$ by swapping of successors as 
explained 
in
the beginning of Section~\ref{sec:idea}.
The restriction $e_i'\neqd e_j'$ reflects the fact
that we have to choose exactly one common node per tour for merging,
as already explained in Section~\ref{sec:idea}, see Figure~\ref{fig:WrongSwap}. 

\begin{lemma}\label{mainLemma}
  Let $k\in\{0,\ldots,N\}$.
  Then, $\delk$ is bijective and for any $(u,v),(u',v') \linebreak \in R^*(E)$, 
we have
  \begin{itemize}
  \item[(i)] If $(u,v),(u',v')$ are processed edges,
    then $(u,v)\eqk(u',v')\Leftrightarrow t_k(u)=t_k(u')$.
  \item[(ii)] If $(u,v)$ is a processed edge, then $t_k(u)=t_k(v)$. 
  \item[(iii)] If $t_k(u)=0$, then $(u,v)\eqk(u',v')\Leftrightarrow 
(u,v)\equiv_{\delta^c}(u',v')$.
  \end{itemize}
\end{lemma}

Claim (i) says that the procedure \textsc{Update} works correctly,
i.e., that the $t$-value of a node (if it isn't $0$)
always represents the tour it currently
is associated to.
Claim (ii) says that after an edge has been processed,
both of their nodes are associated to the same tour.
So after the algorithm has finished,
every node of $G$ is in the same tour as its neighbor.

\begin{proof}
  We prove all claims via one induction over $k$.
  For $k=0$ we have $\delta^*_0=\delta^c$
  which is bijective (Lemma~\ref{lem:delc}).
  Moreover, no edge has been processed so far,
  so (i) and (ii) are trivially fulfilled
  and (iii) follows directly from $\delta^*_0=\delta^c$.
  \par
  Now let all of the claims be true for $k\in \{0,\ldots,N-1\}$.
  We start with proving the bijectivity and (i) for $k+1$.\par
  For this we take a closer look at the $(k+1)$-th call of \textsc{Merge-Cycle}.
  
  If $\delk\neq \delkp$,
  this change has to be happening in one of the procedures
  \textsc{New-Nodes}, \textsc{Merge} or \textsc{Write},
  since these are the only procedures in which
  edges are written to the stream or added to $F$.
  First, note that for every edge $e$
  written to the stream during \textsc{Write}
  or added to $F$ in \textsc{New-Nodes} it holds
  $\delk(e)=\delkp(e)$:\par
  If $e=(v_i,v_{i+1})$ is written to the stream during \textsc{Write},
  it is written in the form $(v_i,v_{i+1},v_{i+2})$,
  so we have
  $\delkp(e)=(v_{i+1},v_{i+2})=\delta^c(e)=\delk(e)$.\par
  If $e=(v_{i-1},v_i)$ is added to $F$ during \textsc{New-Nodes},
  it becomes a type-B-edge at this point,
  so $\delkp(e)=(v_i,j(v_i))$.
  Moreover, $j(v_i)$ is set to $v_{i+1}$ in line $3$,
  so $\delkp(e)=(v_i,v_{i+1})=\delta^c(e)=\delk(e)$.\par
  
  So we may concentrate on the procedure \textsc{Merge}:
  Here we process every node from the set $J_{k+1}$.
  Let $r:=|J_{k+1}|$, for instance $J=\{w_1,\ldots,w_r\}$.
  Each of these nodes $w_i$ has been processed before,
  hence, there is a unique edge in $F_k$ that ends in $w_i$
  and which we denote by $e_i$.
  Moreover, there is a unique edge in $C_{k+1}$ that ends in $w_i$
  and which we denote by $e_i'$.
  Now let $i\in[r]$.
  We process $w_i$ in two steps:\par
  Step 1: $(w_i,j(w_i))$ is marked as successor of $e_i'$.
  So directly after this step, $e_i'$ and $e_i$ share the same successor,
  while the out-going edge of $w_i$ in $C_{k+1}$ has lost its predecessor.\par
  Step 2: $j(w_i)$ is set to the next node in the cycle,
  so that the out-going edge of $w_i$  becomes the successor of $e_i$.
  \par
  In these two steps we swapped the successors of $e_i$ and $e_i'$
  and did not change anything else, so what we get is
  \[\delkp(e)=\delk(e)\text{ for any } e\in \vE\setminus\{e_1,\ldots, 
e_r,e_1',\ldots,e_r'\}\]
  and for any $i\in[r]$ 
  \[\delkp(e_i)=\delk(e_i')\text{ and }\delkp(e_i')=\delk(e_i).\]
  Let $i,j\in[r]$ with $i\neq j$. We have
  $e_i'\eqk e_j'$, because $e_j'\in 
C_{k+1}={[e_i']}_{\delta^c}={[e_i']}_{\delk}$.
  We also have $e_i\neqk e_j$, which follows from $t_k(w_i)\neq t_k(w_j)$ 
(\textsc{Construct-J-M}, line $4$)
  together with the induction hypothesis.
  Finally we have $e_i\neqk e_i'$, because $e_i'\notin 
E(C_{k+1})={[e_i]}_{\delta^c}={[e_i]}_{\delk}$.

  So we can apply Lemma~\ref{technical2} with $\delta=\delk$ and 
$\delta'=\delkp$
  and get the bijectivity of $\delkp$ and for every processed edge $e$
  \begin{align*}
    e\in{[e_1]}_{\delkp}
    \Leftrightarrow 
e\in\bigcup\limits_{i=1}^{r}{[e_i']}_{\delk}\cup{[e_1]}_{\delk}
    \Leftrightarrow t_k(e_{(1)})\in M_k \lor e\in C_{k+1}
    \Leftrightarrow t_{k+1}(e_{(1)})= a_{k+1} 
  \end{align*}
  and
  \begin{align*}
    e\notin{[e_1]}_{\delkp}
    &\Leftrightarrow 
e\notin\bigcup\limits_{i=1}^{r}{[e_i']}_{\delk}\cup{[e_1]}_{\delk}
      \Leftrightarrow t_k(e_{(1)})\notin M_k \land e\notin C_{k+1}\\
    &\Leftrightarrow t_{k+1}(e_{(1)})=t_k(e_{(1)})\neq a_{k+1}.
  \end{align*}
  % where $t_k(e_{(1)})$ is the value of $t(e_{(1)})$ at the end of the $k$-th 
  %call of \textsc{Merge-Cycle}.
  Now we are able to complete the proof of (i):
  Let $(u,v),(u',v')$ be processed edges.\par
  Case 1: $(u,v),(u',v')\in {[e_1]}_{\delkp}$.
  Then $(u,v)\eqkp (u',v')$ and $t(u)=a_{k+1}=t(u')$.\par
  Case 2: $(u,v) \in {[e_1]}_{\delkp},(u',v')\notin {[e_1]}_{\delkp}$.
  Then $(u,v)\neqkp (u',v')$ and $t(u)=a_{k+1}\neq t(u')$.\par
  Case 3: $(u,v) \notin {[e_1]}_{\delkp},(u',v')\in {[e_1]}_{\delkp}$.
  Analog to case 2.\par
  Case 4: $(u,v),(u',v')\notin {[e_1]}_{\delkp}$.
  Then $t_{k+1}(u)=t_k(u),t_{k+1}(u')=t_k(u')$ and (\ref{prop2} of 
Lemma~\ref{technical2})
  ${[(u,v)]}_{\delkp}={[(u,v)]}_{\delk}$ and 
${[(u',v')]}_{\delkp}={[(u',v')]}_{\delk}$.
  So we have 
  \[(u,v)\eqkp(u',v')\Leftrightarrow (u,v)\eqk(u',v')\Leftrightarrow 
t_k(u)=t_k(u')\Leftrightarrow t_{k+1}(u)=t_{k+1}(u').\]

  \par About (ii). Let $(u,v)$ be a processed edge.
  If $(u,v)\in C_{k+1}$,
  then at the end of \textsc{Merge-Cycle} both $t(u)$ and $t(v)$
  are set to the same value $a$.
  If $(u,v)\notin C_{k+1}$, then $(u,v)$ already was a processed edge before
  so by induction hypothesis and Lemma~\ref{obs:obs} we are finished.
  \par About (iii). Let $u\in V$ with $t_{k+1}(u)=0$.
  That means that $u$ is not processed in the first $k+1$ calls of 
\textsc{Merge-Cycle}.
  Especially we have
  $(u,v)\equiv_{\delkp}(u',v')\Leftrightarrow (u,v)\eqk(u',v')\Leftrightarrow 
(u,v)\equiv_{\delta^c}(u',v')$
  by induction hypothesis.  
\end{proof}

These results suffice to proof our main result,
given in the following.

\begin{theorem}\label{finalthm}
  If $G$ is Eulerian, $\delta^*$ determines an Euler tour on $G$.
\end{theorem}
\begin{proof}
  According to Theorem~\ref{thm:classifyEulerTour}, it suffices to show
  that $\delta^*$ is bijective
  and that $e\equiv_{\delta^*}e'$ for any $e,e'\in R^*(E)$.
  Remember that $\delta^*=\delta^*_N$,
  so by  Lemma~\ref{mainLemma} $\delta^*$ is bijective.
  % The second claim follows directly from \autoref{mainLemma} (i) and (ii)
  % together with the fact that $G$ is connected.
  For the second property, let $e,e'\in R^*(E)$
  with $e=(u,v)$ and $e'=(u',v')$.
  If $G$ is Eulerian, it is connected,
  so there exists a $u$-$u'$-path $P$ in $G$.
  For every edge on $P$, either the edge itself
  or the corresponding reversed edge has been processed
  during the algorithm \textsc{ Euler-Tour}.
  By Lemma~\ref{mainLemma} (ii), $t_N(x)=t_N(y)$ for all nodes $x,y$ of $P$,
  hence, $t_N(u)=t_N(u')$
  and by Lemma~\ref{mainLemma} (i), we get $e\equiv_{\delta^*_N}e'$.
  Since $\delta^*_N=\delta^*$, we are done.
\end{proof}

\bibliographystyle{plain}
% \bibliography{test.bib}
\bibliography{literature-ets}

\newpage

\section*{Appendix}

\begin{wrapfigure}{r}[-0.3cm]{0.3\textwidth}%[6cm]{}
  \fbox{
    \includegraphics[width=0.9\linewidth]{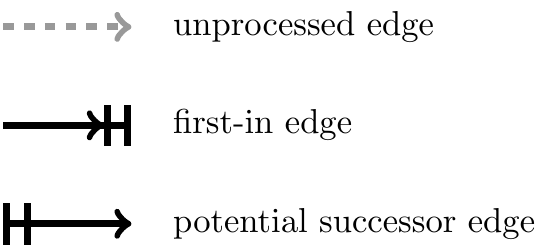}
  }
\end{wrapfigure}
On the following two pages
we present a working example for the method \textsc{Merge-Cycle}
that corresponds to the steps~\ref{step2}\ to~\ref{step6}\ in our
high level description.
Note that every node has at most one in-going first-in edge
and one out-going potential successor edge at a time.\\
 \capstartfalse
\begin{figure}[h]
  \begin{subfigure}{0.5\textwidth}
    \captionsetup{width=0.9\linewidth}
    \fbox{
      \includegraphics[width=0.9\linewidth]{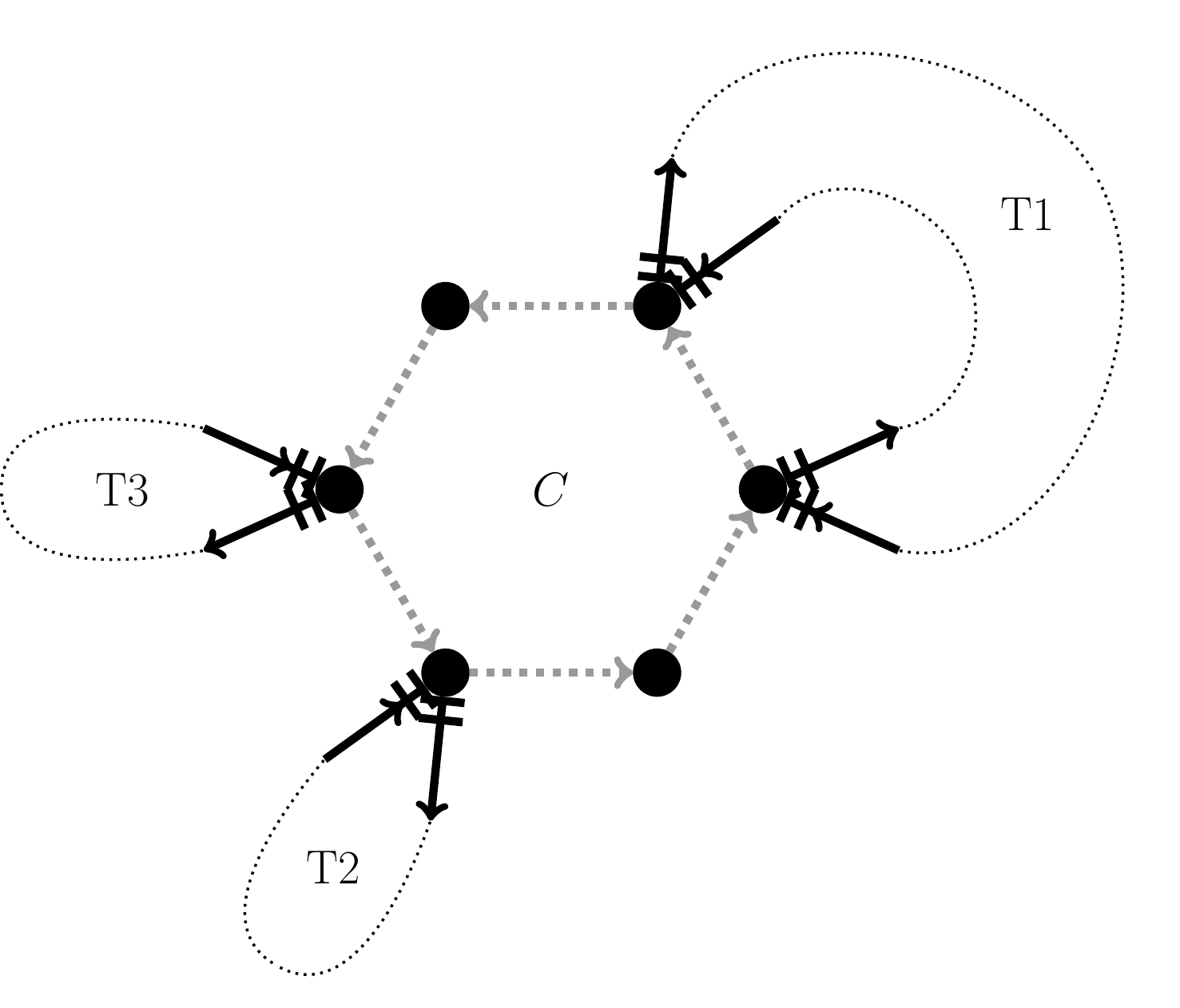}
    }
    \caption*{A cycle $C$ has been found.\newline \newline }
  \end{subfigure}
  \begin{subfigure}{0.5\textwidth}
    \captionsetup{width=0.9\linewidth}
    \fbox{
      \includegraphics[width=0.9\linewidth]{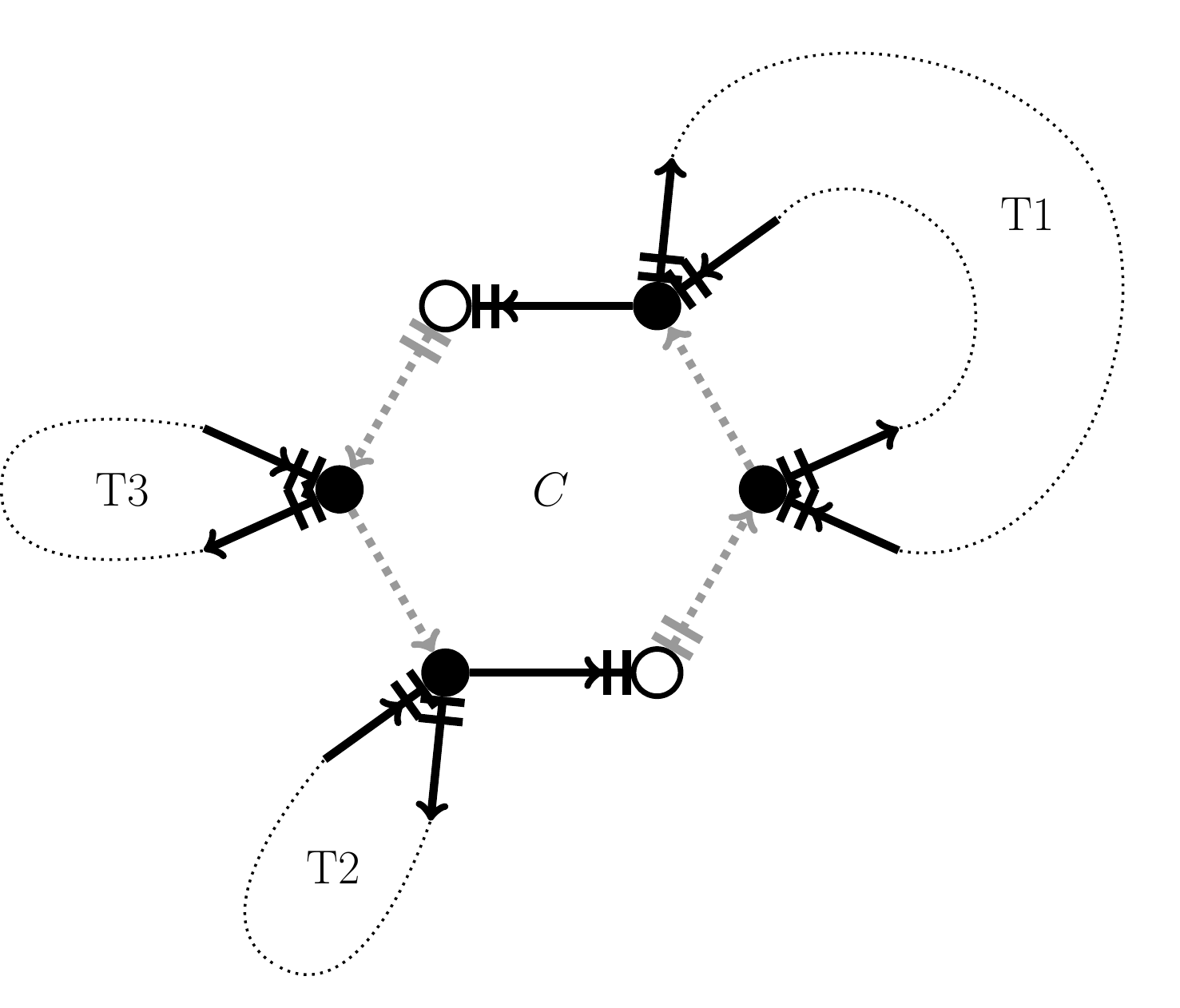}
    }
    \caption*{Step~\ref{step2}\ For every new node the in-going edge becomes a 
first-in edge and the outgoing edge becomes a potential successor.}
  \end{subfigure}
\end{figure}

\begin{figure}[h]
  \begin{subfigure}{0.5 \textwidth}
    \captionsetup{width=0.9\linewidth}
    \fbox{
      \includegraphics[width=0.9\linewidth]{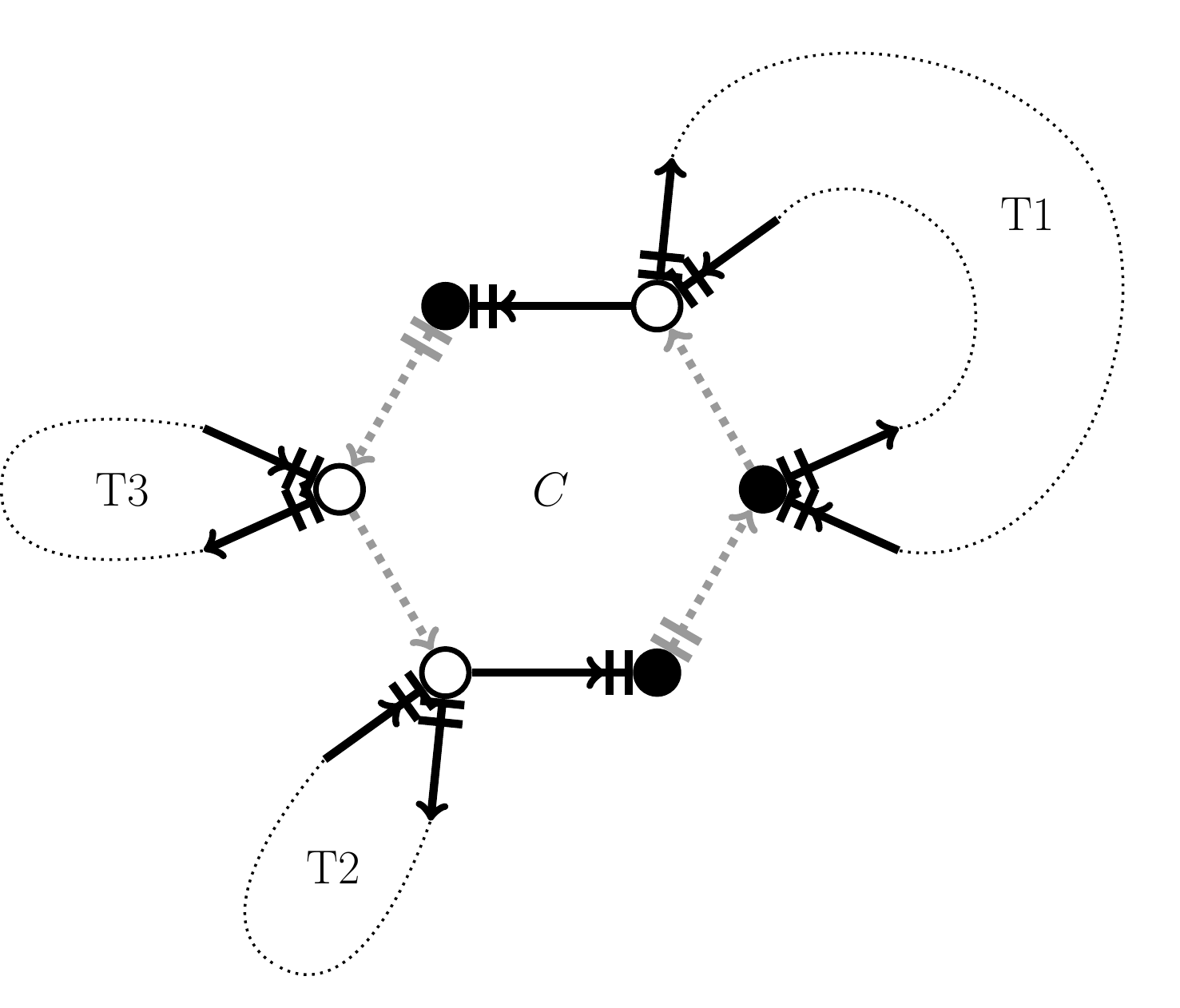}
    }
    \caption*{Step~\ref{step3}\ For each intersecting Tour $T_1,T_2,T_3$
      we choose one common node.\newline }
  \end{subfigure}
  \begin{subfigure}{0.5\textwidth}
    \captionsetup{width=0.9\linewidth}
    \fbox{
      \includegraphics[width=0.9\linewidth]{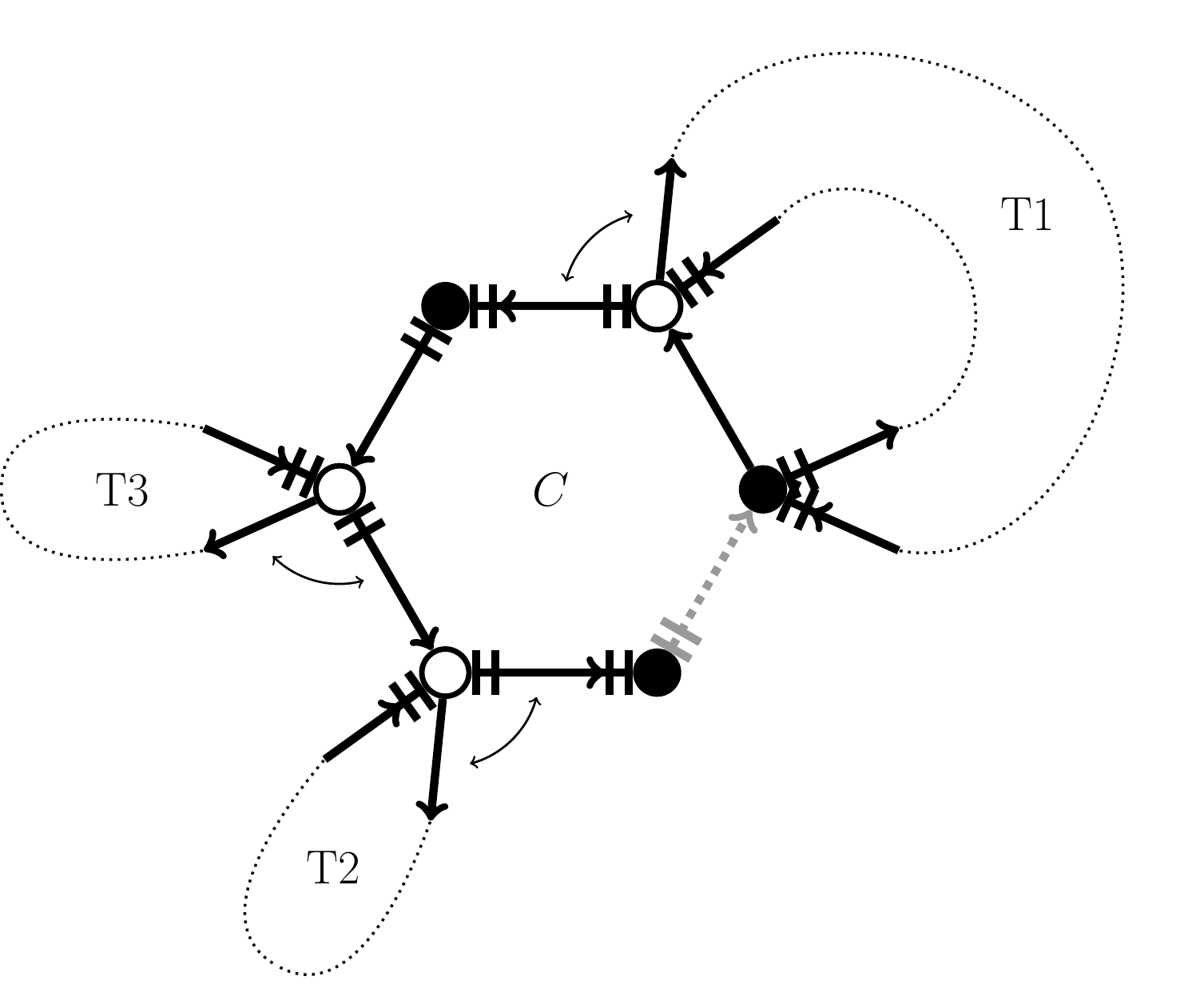}
    }
    \caption*{Step~\ref{step4}\ The successors of the chosen nodes are swapped 
with potential successor edges.}
  \end{subfigure}
\end{figure}

\begin{figure}[h]
  \begin{subfigure}{0.5 \textwidth}
    \captionsetup{width=0.9\linewidth}
    \fbox{
      \includegraphics[width=0.9\linewidth]{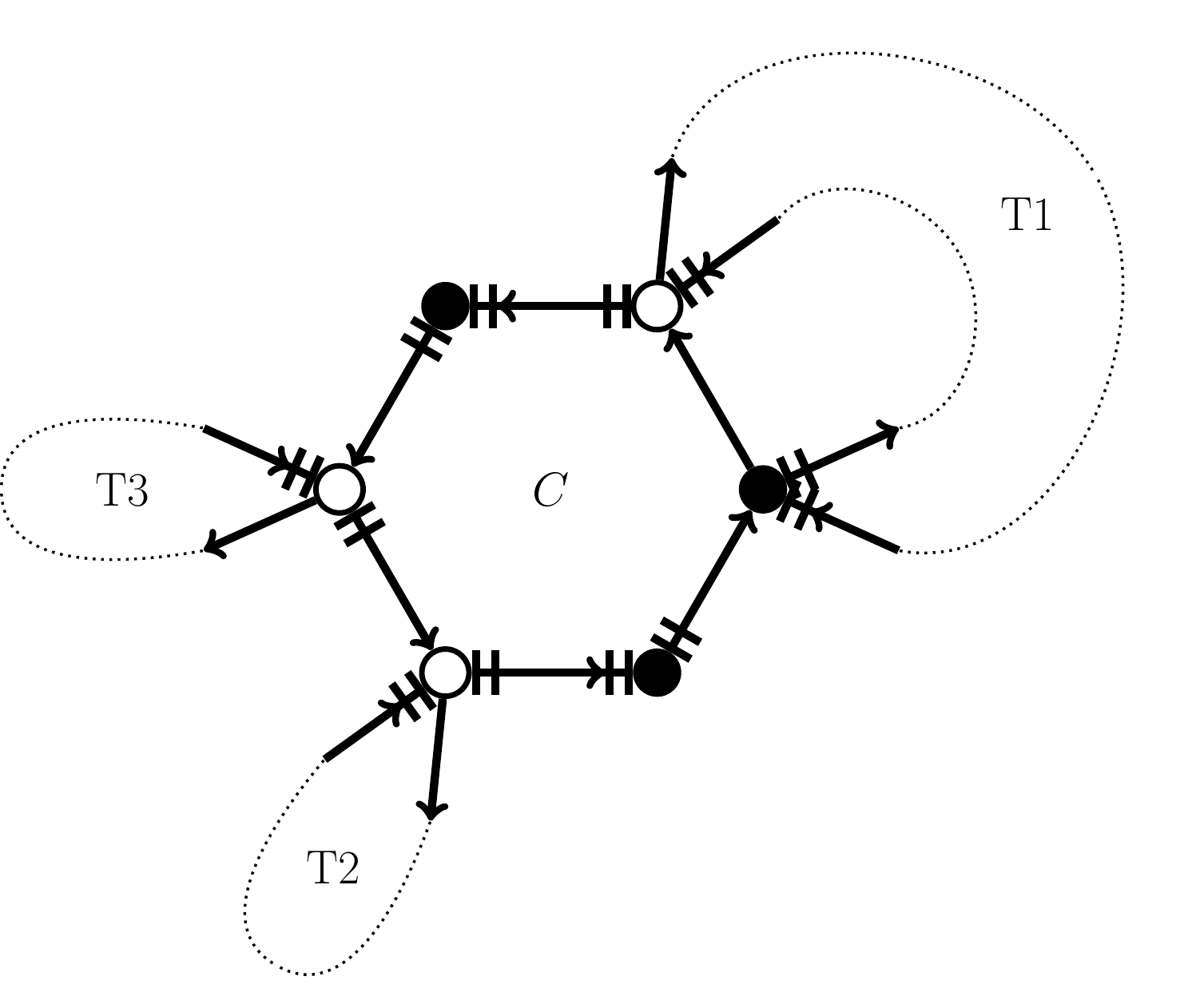}
    }
    \caption*{Step~\ref{step5}\ For the rest of the edges the successor stays 
the same.}
  \end{subfigure}
  \begin{subfigure}{0.5\textwidth}
    \captionsetup{width=0.9\linewidth}
    \fbox{
      \includegraphics[width=0.9\linewidth]{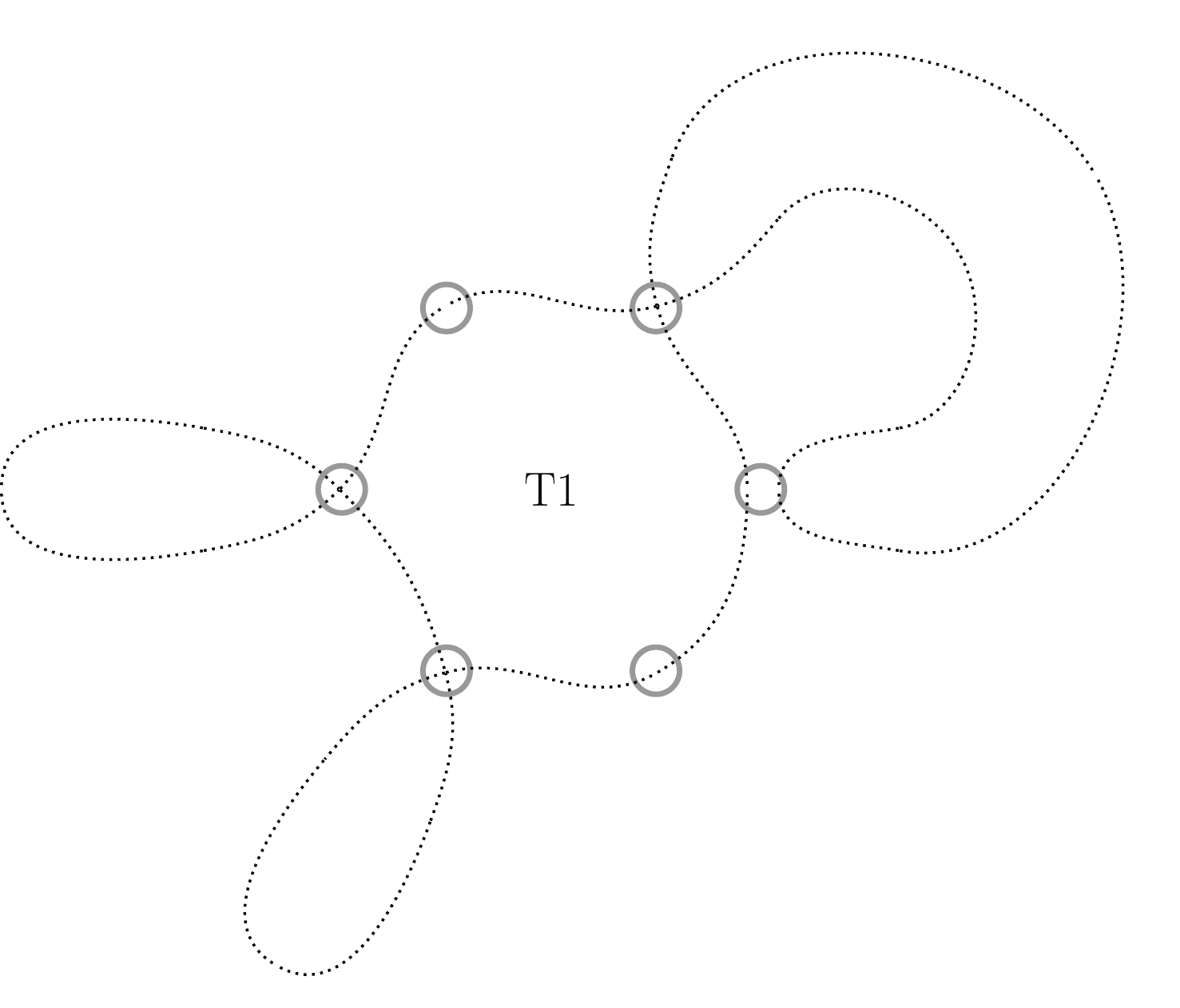}
    }
    \caption*{Step~\ref{step6}\ The tours and the cycle have been merged to one 
tour.}
  \end{subfigure}
\end{figure}

\end{document}